\documentclass[conference]{IEEEtran}
\ifCLASSINFOpdf
\else
\fi

\usepackage{cite}

\usepackage[cmex10]{amsmath}
\usepackage{microtype}
\usepackage{graphicx}
\usepackage{subdepth}  
\usepackage{bm} 
\usepackage{booktabs} 

\usepackage{subcaption} 

\usepackage{tikz}
\usetikzlibrary{positioning}
\usetikzlibrary{shadows}
\usepackage{pgfmath}

\usepackage{multirow}

\usepackage{capt-of}

\tikzset{%
  every neuron/.style={
    circle,
    draw,
    minimum size=1cm
  },
  neuron missing/.style={
    draw=none,
    fill=none,
    yshift = 0.1cm,
    execute at begin node=\color{black}$\vdots$
  },
  neuron nothing/.style={
    draw=none,
    fill=none,
  },  
}
\tikzstyle{bias neuron}=[neuron, minimum size=12pt, line width = .5, opacity = 0.6]
\tikzstyle{rect neuron}=[line width = .8, draw = black, rectangle, text width = 4em, minimum height = 2em, text centered, scale = .5,inner sep=0pt]

\usepackage{hyperref}

\hypersetup{colorlinks = true, citecolor = blue}

\usepackage{algorithm}
\usepackage{algorithmic}

\usepackage[mathcal]{euscript} 
\usepackage[cmex10]{amsmath}
\usepackage{amsfonts,amssymb,amsthm,bbm}

\DeclareMathAlphabet{\mathbbb}{U}{bbold}{m}{n}  

\newtheorem{theorem}{Theorem}
\newtheorem{lemma}{Lemma}
\newtheorem{proposition}{Proposition}
\newtheorem{corollary}{Corollary}
\newtheorem{definition}{Definition}
\newtheorem{example}{Example}
\newtheorem{remark}{Remark}

\newcommand{\secref}[1]{Section~\ref{#1}}
\newcommand{\appref}[1]{Appendix~\ref{#1}}

\newcommand{\figref}[1]{Fig.~\ref{#1}}
\newcommand{\lemref}[1]{Lemma~\ref{#1}}
\newcommand{\thmref}[1]{Theorem~\ref{#1}}

\DeclareMathOperator*{\argmax}{arg\,max}
\DeclareMathOperator*{\argmin}{arg\,min}
\DeclareMathOperator*{\relint}{relint}
\DeclareMathOperator*{\tr}{tr}

\DeclareMathOperator*{\diag}{diag}

\newcommand{\kron}{\mathbbb{1}}

\newcommand{\p}{\partial}
\newcommand{\eps}{\epsilon}

\newcommand{\la}{\lambda}

\newcommand{\bA}{\mathbf{A}}

\newcommand{\bU}{\mathbf{U}}
\newcommand{\bV}{\mathbf{V}}

\newcommand{\bz}{{\mathbf{z}}}
\newcommand{\C}{{\mathcal{C}}}

\newcommand{\U}{{\mathcal{U}}}
\newcommand{\V}{{\mathcal{V}}}

\newcommand{\X}{{\mathcal{X}}}

\newcommand{\Y}{{\mathcal{Y}}}

\newcommand{\Z}{{\mathcal{Z}}}

\newcommand{\cN}{\mathcal{N}}
\newcommand{\cP}{\mathcal{P}}

\newcommand{\cS}{\mathcal{S}}
\newcommand{\cI}{\mathcal{I}}

\newcommand{\cU}{\U}
\newcommand{\cV}{\V}

\newcommand{\cX}{\X}

\newcommand{\cY}{\Y}

\newcommand{\fh}{\hat{f}}

\newcommand{\gh}{\hat{g}}

\newcommand{\ev}{{\mathbb{E}}}

\newcommand{\E}[1]{\ev\left[#1\right]}

\newcommand{\bE}[1]{\ev\bigl[#1\bigr]}

\newcommand{\Ed}[2]{\ev_{#1}\left[#2\right]}
\newcommand{\bEd}[2]{\ev_{#1}\bigl[#2\bigr]}

\DeclareMathOperator{\var}{var}

\newcommand{\expop}[1]{\exp\left\{#1\right\}}
\newcommand{\reals}{\mathbb{R}}

\newcommand{\norm}[1]{\|#1\|}

\newcommand{\ip}[2]{\langle#1,#2\rangle}
\newcommand{\bip}[2]{\bigl\langle#1,#2\bigr\rangle}

\newcommand{\bxi}{\boldsymbol{\xi}}

\newcommand{\bXi}{\boldsymbol{\Xi}}

\newcommand{\bpsi}{\boldsymbol{\psi}}

\newcommand{\bphi}{\boldsymbol{\phi}}

\newcommand{\frob}[1]{\|#1\|_\mathrm{F}}
\newcommand{\bfrob}[1]{\bigl\|#1\bigr\|_\mathrm{F}}
\newcommand{\bbfrob}[1]{\left\|#1\right\|_\mathrm{F}}

\newcommand{\defeq}{\triangleq}
\newcommand{\simpXY}{\cP^{\X\times\Y}}

\newcommand{\simpX}{\cP^\X}
\newcommand{\simpY}{\cP^\Y}

\newcommand{\nbhd}{\cN}

\newcommand{\Ph}{\hat{P}}
\newcommand{\empgen}{\Ph} 
\newcommand{\refgen}{P_0} 

\newcommand{\fvgen}{\xi}

\newcommand{\ivgen}{\phi}

\newcommand{\ivemp}{\hat{\phi}}  
\newcommand{\ivspacegen}{\cI^\Z} 
\newcommand{\ivspaceX}{\cI^\X}   
\newcommand{\ivspaceY}{\cI^\Y}   
\newcommand{\llgen}{L}

\newcommand{\dtmh}{\mathbf{B}}
\newcommand{\dtm}{\tilde{\dtmh}}

\newcommand{\bP}{\mathbf{P}}

\newcommand{\bQ}{\mathbf{Q}}
\newcommand{\eqd}{\stackrel{\mathrm{d}}{=}}

\newcommand\T{\mathrm{T}}

\newcommand{\bLa}{\boldsymbol{\Lambda}}

\newcommand{\eff}{\nu} 
\newcommand{\Pt}{\tilde{P}}
\newcommand{\phit}{\tilde{\phi}}

\newcommand{\bone}{\mathbf{1}}
\newcommand{\bI}{\mathbf{I}}

\definecolor{electricpurple}{rgb}{0.75, 0.0, 1.0}

\newcommand{\rep}{s}
\newcommand{\weight}{v}
\newcommand{\bias}{b}
\newcommand{\rinf}{\xi^X}
\newcommand{\winf}{\xi^Y}
\newcommand{\rmat}{\bXi^X}
\newcommand{\wmat}{\bXi^Y}

\newcommand{\hdtm}{\dtm_1}

\newcommand{\repm}{\mu_s}        

\newcommand{\hrep}{\rep}
\newcommand{\hweight}{\weight}
\newcommand{\hbias}{\bias}
\newcommand{\hrmat}{\rmat}
\newcommand{\hwmat}{\wmat}

\newcommand{\bhrep}{t}
\newcommand{\bhweight}{w}
\newcommand{\bhbias}{c}
\newcommand{\bhrmat}{\hrmat_1}

\newcommand{\hrepm}{\repm}        
\newcommand{\hrepp}{\repp} 

\newcommand{\repp}{\tilde{\rep}} 
\newcommand{\weightp}{\tilde{\weight}}
\newcommand{\biasd}{d}
\newcommand{\biasdp}{\tilde{\biasd}}
\newcommand{\hweightp}{\tilde{\hweight}}

\newcommand{\bhrepp}{\tilde{\bhrep}}

\newcommand{\prior}{\Theta}

\newcommand{\lossb}{\eta}
\newcommand{\loss}{{\cal L}}
\newcommand{\lossbh}{\kappa}

\newcommand{\rdim}{k} 
\newcommand{\hrdim}{k} 
\newcommand{\bhrdim}{m} 
\newcommand{\ydim}{|\mathcal{Y}|} 

\newcommand{\bhrepm}{\mu_{{\bhrep}}}

\newcommand{\eqspace}{\quad\,\,}
\newcommand{\hrepsm}{\mu_{s^*}}

\newcommand{\markov}{\leftrightarrow}

\allowdisplaybreaks

\title{An Information Theoretic Interpretation to Deep Neural Networks}

\author{
  \IEEEauthorblockN{Shao-Lun Huang}
  \IEEEauthorblockA{DSIT Research Center\\
    Tsinghua-Berkeley Shenzhen Institute\\
    Shenzhen, China 518055\\
    Email: shaolun.huang@sz.tsinghua.edu.cn} 
  \and
\IEEEauthorblockN{Xiangxiang Xu}
  \IEEEauthorblockA{Dept. of Electronic Engineering\\
    Tsinghua University\\
    Beijing, China 100084\\
    Email: xuxx14@mails.tsinghua.edu.cn} 
  \and
  \IEEEauthorblockN{Lizhong Zheng, Gregory W. Wornell}
  \IEEEauthorblockA{Dep. of Electrical \& Computer Eng.\\
    Massachusetts Institute of Technology\\
    Cambridge, MA 02139-4307\\
    Email: \{lizhong, gww\}@mit.edu}
}

\begin{document}

\maketitle

\begin{abstract}
It is commonly believed that the hidden layers of deep neural networks (DNNs) attempt to extract informative features for learning tasks. In this paper, we formalize this intuition by showing that the features extracted by DNN coincide with the result of an optimization problem, which we call the ``universal feature selection'' problem, in a local analysis regime. We interpret the weights training in DNN as the projection of feature functions between feature spaces, specified by the network structure. Our formulation has direct operational meaning in terms of the performance for inference tasks, and gives interpretations to the internal computation results of DNNs. Results of numerical experiments are provided to support the analysis.
\end{abstract}

\IEEEpeerreviewmaketitle

\section{Introduction}

Due to the striking performance of deep learning in various fields, deep neural networks (DNNs) have gained great attentions in modern computer science. While it is a common understanding that the features extracted from the hidden layers of DNN are ``informative" for learning tasks, the mathematical meaning of informative features in DNN is generally not clear. There have been numerous research efforts towards this direction~\cite{MacKay03}. For instance, the information bottleneck~\cite{tishby2015deep} employs the mutual information as the metric to quantify the informativeness of features in DNN, and other information metrics, such as the Kullback-Leibler (K-L) divergence~\cite{huang2017information} and Weissenstein distance~\cite{arjovsky2017wasserstein} are also used in different problems. However, because of the complicated structure of DNNs, there is a disconnection between these information metrics and the performance objectives of the inference tasks that DNNs want to solve. Therefore, it is in general difficult to match the DNN learning with the optimization of a particular information metric.



In this paper, our first contribution is to propose a learning framework, called universal feature selection, which connects the information metric of features and the performance evaluation of inference problems. Specifically for a pair of data variables $X$ and $Y$, the goal of universal feature selection is to select features from $X$ to infer about a targeted attribute $V$ of $Y$, where $V$ is only assumed with a rotationally uniform prior over the attribute space of $Y$, but the precise statistical model between $V$ and $X$ is unknown. Thus, the selected features have to be good for solving multiple inference problems, and should be generally ``informative" about $Y$. We show that in a local analysis regime, the averaged performance of inferring $V$ by a selected feature of $X$ is measured via a linear projection of this feature, which leads to an information metric to features, and the optimal features can be computed from the singular value decomposition (SVD) of this linear projection.



More importantly, we show that in the local analysis regime, the optimal features selected in DNNs from log-loss optimization coincide with the solutions of universal feature selection. Therefore, the information metric developed in universal feature selection can be used to understand the operations in DNNs. As a result, we observe that the DNN weight updates in general can be interpreted as projecting features between the feature spaces of data and label for extracting the most correlated aspects between them, and the iterative projections can be viewed as computing the SVD of a linear projection between these feature spaces. Moreover, our results also give an explicit interpretation of the goal and the procedures of the BackProp/SGD operations in deep learning. Finally, the theoretic results are validated via numerical experiments.

 \textbf{Notations:}
Throughout this paper, we use $X$, $\X$, $P_X$, and $x$ to represent a discrete random variable, the range, the probability distribution, and the value of $X$. In addition, for any function $\rep(X) \in \mathbb{R}^k$ of $X$, we use $\repm$ to denote the mean of $\rep(X)$, and ``$\tilde{~}$'' to denote the mean removed version of a variable; e.g., $\repp(X) = \rep(X) - \repm$. Finally, we use $\| \cdot \|$ and $\| \cdot \|_{\text{F}}$ to denote the $\ell_2$-norm 
and the Frobenius norm, 
respectively.


\section{Preliminary and Definition} \label{sec:PD}



Given a pair of discrete random variables $X,Y$ with the joint distribution $P_{XY}(x,y)$, the $|\Y|\times|\X|$ matrix $\dtm$ is defined as
\begin{equation} \label{eq:dtm-def}
\dtm (y,x) \defeq \frac{P_{XY}(x,y) - P_X(x) P_Y(y)}{\sqrt{P_X(x) P_Y(y)}},
\end{equation}
where $\dtm (y,x)$ is the $(y,x)$th entry of $\dtm$. The matrix $\dtm$ is referred to as the canonical dependence matrix (CDM). The SVD of $\dtm$ has the following properties \cite{huang2017information}.

\begin{lemma}
  The SVD of $\dtm$ can be written as $\dtm = \sum_{i=1}^{K} \sigma_i\, \bpsi^Y_i \bigl(\bpsi^X_i\bigr)^\T$, where $K\defeq\min\{|\X|,|\Y|\}$, and $\sigma_i$ denotes the $i$\/th singular value with the ordering $1 \geq \sigma_1 \ge \dots \ge \sigma_{K} = 0$, and $\bpsi^Y_i$ and $\bpsi^X_i$ are the corresponding left and right singular vectors with $\psi^X_K(x) = \sqrt{P_X(x)}$ and $\psi^Y_K(y) = \sqrt{P_Y(y)}$.
  \label{lem:dtm:svd}
\end{lemma}


This SVD decomposes the feature spaces of $X,Y$ into maximally correlated features. To see that, consider the generalized canonical correlation analysis (CCA) problem:
\begin{align*} 
\max_{\substack{\E{f_i(X)} = \E{g_i(Y)} = 0 \\ {\E{f_i(X)\,f_j(X)}  = \E{g_i(Y)\,g_j(Y)} =  \kron_{i=j}}}} \sum_{i=1}^k \E{f_i(X)\,g_i(Y)}.
\end{align*}
It can be shown that for any $1 \leq k \leq K-1$, the optimal features are $f_i(x) = \psi^X_i(x) / \sqrt{P_X(x)}$, and $ g_i(y) = \psi^Y_i(y) /\sqrt{P_Y(y)}$, for $i=0,\dots,K-1$, where $\psi^X_i(x)$ and $\psi^Y_i(y)$ are the $x$\/th and $y$\/th entries of $\bpsi^X_i$ and $\bpsi^Y_i$, respectively~\cite{huang2017information}. The special case $k=1$ corresponds to the HGR maximal correlation~\cite{hoh35,hg41,ar59}, and the optimal features can be computed from the ACE algorithm~\cite{Breiman85}.

Moreover, in this paper we focus on a particular analysis regime described as follows.

\begin{definition}[$\eps$-Neighborhood]
Let $\simpX$ denote the space of distributions
on some finite alphabet $\X$, and let $\relint(\simpX)$ denote the
subset of strictly positive distributions. 
For a given $\eps>0$, the $\eps$-neighborhood of a
distribution $P_X \in \relint(\simpX)$ is defined by the $\chi^2$-divergence as
\begin{equation*}
\nbhd_\eps^\X(P_X) 
\defeq \left\{ P \in \simpX \colon
\sum_{x\in\X} \frac{\bigl(P(x)- P_X(x)\bigr)^2}{P_X(x)} \le \eps^2 \right\}.
\end{equation*}
\end{definition}
\begin{definition}[$\eps$-Dependence]
The random variables $X,Y$ is called $\eps$-dependent if $P_{XY} \in \nbhd_\eps^{\X \times \Y}(P_XP_Y)$.
\end{definition}
\begin{definition}[$\eps$-Attribute]
A random variable $U$ is called an $\eps$-attribute of $X$ if $P_{X|U} (\cdot | u) \in \nbhd_\eps^{\X }(P_X)$, for all $u \in \cU$.
\end{definition}
Throughout this paper, we focus on the small $\eps$ regime, which we refer to as the local analysis regime. 
In addition, for any $P \in \simpX$, we define the \emph{information vector} $\phi$ and \emph{feature function} $\llgen(x)$ corresponding to $P$, with respect to a reference distribution $P_X \in \relint(\simpX)$, as
\begin{equation}
\phi(x) \defeq  \frac{P(x) -
  P_X(x)}{\sqrt{P_X(x)}}, \quad \llgen(x) \defeq \frac{\phi(x)}{\sqrt{P_X(x)}}.
\label{eq:iv-def}
\end{equation}
This gives a three way correspondence $P \leftrightarrow \phi \leftrightarrow L$ for all distributions in $\nbhd_\eps^\X(P_X)$, which will be useful in our derivations.

\section{Universal Feature Selection}
\label{sec:ufs}
Suppose that given random variables $X,Y$ with joint distribution $P_{XY}$, we want to infer about an attribute $V$ of $Y$ from observed i.i.d. samples $x_1, \dots , x_n$ of $X$. When the statistical model $P_{X|V}$ is known, the optimal decision rule is the log-likelihood ratio test, where the log-likelihood function can be viewed as the optimal feature for inference. However, in many practical situations~\cite{huang2017information}, it is hard to identify the model of the targeted attribute, and is necessary to select low-dimensional informative features of $X$ for inference tasks before knowing the model. 
We call this universal feature selection problem. 
To formalize this problem, for an attribute $V$, we refer to $\C_\Y = \bigl\{\, \V,\ \{P_V(v),\ v\in\V\},\ \{\bphi^{Y|V}_v,\ v\in\V\} \bigr\},$
as the \emph{configuration} of $V$, where $\phi^{Y|V}_v \leftrightarrow P_{Y|V}(\cdot|v)$ is the information vector specifying the corresponding conditional distribution $P_{Y|V}(\cdot|v)$. The configuration of $V$ models the statistical correlation between $V$ and $Y$. In the sequel, we focus on the local analysis regime, for which we assume that all the attributes $V$ of our interests to detect are $\eps$-attributes of $Y$. As a result, the corresponding configuration satisfies $\| \phi^{Y|V}_v \| \leq \eps$, for all $v \in \V$. We refer to this as the \emph{$\eps$-configurations}. The configuration of $V$ is unknown in advance, but assumed to be generated from a \emph{rotational invariant ensemble (RIE)}.
\begin{definition}[RIE]
Two configurations $\C_\Y $ and $\tilde{\C}_\Y$ defined as
$$\C_\Y = \bigl\{\, \V,\ \{P_V(v),\ v\in\V\},\ \{\bphi^{Y|V}_v,\ v\in\V\} \bigr\},$$
$$\tilde{\C}_\Y \defeq \bigl\{\, \V,\ \{P_V(v),\ v\in\V\},\ \{\tilde{\bphi}\vphantom{\bphi}^{Y|V}_v,\ v\in\V\} \bigr\}$$ are called rotationally equivalent, if there exists a unitary matrix $\bQ$ such that $\tilde{\bphi}\vphantom{\bphi}^{Y|V}_v = \bQ \, \bphi^{Y|V}_v$, for all $v \in \V$. Moreover, a probability measure defined on a set of configurations is called an RIE, if all rotationally equivalent configurations have the same measure.
\end{definition}
The RIE can be interpreted as assigning a uniform measure to the attributes with the same level of distinguishability. To infer about the attribute $V$, we construct a $k$-dimensional feature vector $h^k = (h_1,\dots,h_k)$, for some $1\le k\le K-1$, of the form
$h_i = \frac1n \sum_{l=1}^n f_i(x_l), \ i =  1, \ldots , k$,
for some choices of feature functions $f_i$. Our goal is to determine the $f_i$ such that the optimal decision rule based on $h^k$ achieves the smallest possible error probability, where the performance is averaged over the possible $\C_\Y$ generated from an RIE. In turn, we denote $\xi^X_i \leftrightarrow f_i$ as the corresponding information vector, and define the matrix $\bXi^X \defeq [  \bxi^X_1 \ \cdots \  \bxi^X_k ]$.




\begin{theorem}[Universal Feature Selection] \label{thm:UFS}
For $v,v' \in \V$, let $E_{h^k}(v,v')$ be the error exponent associated with the pairwise error probability distinguishing $v$ and $v'$ based on $h^k$, then the expectation of the error exponent over a given RIE defined on the set of $\eps$-configuration is given by
\begin{align} \notag
&\E{ E_{h^k}(v,v') } \\ \label{eq:EEg-k}
&= \frac{\E{\bigl\| \bphi^{Y|V}_v - \bphi^{Y|V}_{v'} \bigr\|^2}}{8|\Y|} \bbfrob{\dtm \bXi^X \bigl(\bigl(\rmat\bigr)^\T  \rmat\bigr)^{-\frac{1}{2}}}^2 + o(\eps^2),
\end{align}
where the expectations are taken over this RIE. 
\end{theorem}
\begin{proof}
See Appendix~\ref{app:1}.
\end{proof}

As a result of~\eqref{eq:EEg-k}, designing the $\bxi^X_i$ as the singular vectors $\bpsi^X_i$ of $ \dtm$, for $i=1,\dots,k$, optimizes~\eqref{eq:EEg-k} for all RIEs, pairs of $(v,v')$, and $\eps$-configurations. Thus, the feature functions corresponding to $\bpsi^X_i$ are \emph{universally optimal} for inferring the unknown attribute $V$. Moreover,~\eqref{eq:EEg-k} naturally leads to an information metric $\frob{\dtm \bXi^X \bigl(\bigl(\rmat\bigr)^\T  \rmat\bigr)^{-\frac{1}{2}}}^2$ for any feature $\bXi^X$ of $X$, measured by projecting the normalized $\bXi^X$ through a linear projection $\dtm$. This information metric quantifies how informative a feature of $X$ is when solving inference problems with respect to $Y$, and is optimized when designing features by singular vectors of $\dtm$. Thus, we can interpret the universal feature selection as solving the most informative features for data inferences via the SVD of $\dtm$, which also coincides with the maximally correlated features in \secref{sec:PD}. Later on we will show that the feature selections in DNN share the same information metric as universal feature selection in the local analysis regime.

\section{Interpreting Softmax Regression} 
\label{sec:softmax}

To begin, recall that for a data vector $X$ and label $Y$ with labeled samples $(x_i, y_i)$, for $i = 1, \ldots , n$, the softmax regression generally uses a discriminative model of the form 
\begin{align} \label{eq:softmax}
\Pt_{Y|X}^{(\weight,\bias)} (y|x)   \defeq \frac{e^{\weight^\T(y)\rep(x) + \bias(y)}}{\sum_{y' \in \Y}e^{\weight^\T(y')\rep(x) + \bias(y')}  }
\end{align}
to address the classification problems, where $s(x) \in \mathbb{R}^k$ is a $k$-dimensional representation of $X$ used to predict the label, and $\weight(y) \in \mathbb{R}^k$ and $\bias(y) \in \mathbb{R}$ are the parameters required to be learned from 
\begin{equation} 
(\weight,\bias)^* = \argmax_{(\weight,\bias)} \frac{1}{N} \sum_{i=1}^N \log
\Pt_{Y|X}^{(\weight,\bias)}(y_i|x_i).
\label{eqn:logistic_opt}
\end{equation}
As depicted in \figref{fig:NN}, the ordinary softmax regression corresponds to $\rep(x) = x$. More generally, $\rep(x)$ can be the output of the previous hidden layer of a neural network, i.e., the selected feature of $x$ fed into the softmax regression.  In the rest of this section, we will show that when $X,Y$ are $\eps$-dependent, the functions $\rep(x)$ and $\weight(y)$ coincide with the solutions of the universal feature selection.


\begin{figure}[t]
  \centering
   \resizebox{\columnwidth} {!}{\def\layersep{2.5cm}


\begin{tikzpicture}[shorten >=1pt,->,draw=white!60!black!20!blue, node distance=\layersep]
  \tikzstyle{every pin edge}=[<-,shorten <=1pt]
  \tikzstyle{neuron}=[line width = .8, draw = black, circle, minimum size=18pt,inner sep=0pt];
  \tikzstyle{cdot}=[draw=none, fill=none, execute at begin node=\color{black}$\cdots$];
  \tikzstyle{input neuron}=[neuron]; 
  \tikzstyle{hidden neuron}=[neuron]; 
  \tikzstyle{output neuron}=[neuron]; 
  \tikzstyle{annot} = [text width=4em, text centered];
  \node[annot, text width = 4em] at (1.3*\layersep,-2.5) {\small{{\sf{Input Features}}}};


  \node[rectangle, draw = black] (A) at (4.5*\layersep,-2.5) {$\displaystyle \Pt_{Y|X}^{(\weight,\bias)} (y|x) \defeq \frac{e^{\weight^{\T}(y)\rep(x) + \bias(y)}}{\sum_{y' \in \Y}e^{\weight^{\T}(y')\rep(x) + \bias(y')}}$};  
  \node[above] at (A.north) {\small{\sf{Softmax Output}}};

\foreach \m [count=\y] in {1, missing, 2}
  \node [hidden neuron/.try, neuron \m/.try, yshift=-1.1cm] (hidden-\m) at (2*\layersep, - 0.7 * \y cm) {};

\foreach \m [count=\y] in {$\hrep_1$, , $\hrep_{\hrdim}$}
  \node [yshift=-1.1cm] at (2*\layersep, -0.7 * \y cm) {\small\m}; 
  
  \node [bias neuron] (hidden-bias) at (2*\layersep, -4 cm) {\scriptsize{+1}};
  
  \foreach \text / \m [count=\y] in {$Y = 1$/1, $Y = 2$/2, /missing, $Y = {\ydim}$/3}
  \node [rect neuron/.try, neuron \m/.try] (output-\m) at (3*\layersep,-\y) {\text}; 
    





\foreach \i in {1,...,2}
  \foreach \j in {1,...,3}
    \draw [->, line width = .6] (hidden-\i) -- (output-\j.west);

\foreach \text / \ys  [count=\i] in {$\hweight(1)$ / .15cm, $\hweight(2)$ / .05 cm, $\hweight(\ydim)$ / -0.18 cm}
   \draw [draw = none] (hidden-missing) -- (output-\i) node[pos=0.7, yshift = \ys]{\scriptsize\text};

\foreach \text [count=\i] in {$\bias(1)$, $\bias(2)$, $\bias(\ydim)$}
\draw [->, line width = .6, opacity = 0.4] (hidden-bias) -- (output-\i.west);

\foreach \text / \xs / \ys [count=\i] in {$\bias(1)$ / .1 / .1cm, $\bias(2)$ / .1 / -.2cm, $\bias(\ydim)$ /.1 / -.11 cm}
   \draw [->, line width = .6, draw = none, opacity = 0.6] (hidden-bias) -- (output-\i) node[pos=0.25, xshift = \xs, yshift = \ys]{\tiny\text};

\foreach \l [count=\i] in {1,2}
  \draw [<-, line width = .6] (hidden-\i) -- ++(-1.2,0)
    node [left] {};

   

\end{tikzpicture}


}
  \caption{A simple neural network with one layer of hidden nodes with softmax output.}
  \label{fig:NN}
\end{figure}
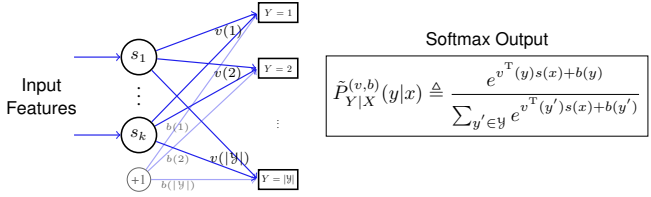



First, we use $P_{XY}$ to denote the joint empirical distribution of the labeled samples $(x_i, y_i), i=1, \dots, N$, and $P_X, P_Y$ to denote the corresponding marginal distributions. Then, the objective function in the optimization problem
\eqref{eqn:logistic_opt} is precisely the empirical average of the log-likelihood, i.e., $\frac{1}{N} \sum_{i=1}^N \log \Pt_{Y|X}^{(\weight,\bias)}(y_i | x_i) = \Ed{P_{XY}}{\log \Pt_{Y|X}^{(\weight,\bias)}(Y|X)}.$
Therefore, maximizing this empirical average is equivalent as minimizing the K-L
divergence:  
\begin{align} 
(\weight,\bias)^* 
= \argmin_{(\weight,\bias)} \quad D( P_{XY} \| P_X \, \Pt^{(\weight,\bias)}_{Y|X}).
\label{eq:KL_softmax} 
\end{align}
This can be interpreted as finding the best fitting to empirical joint distribution $P_{XY}$ by distributions of the form $P_X \, \Pt^{(\weight,\bias)}_{Y|X}$. In our development, it is more convenient to denote the bias by $d(y) = \bias(y) - \log P_Y(y)$, for $y \in \cY$.
Then, the following lemma illustrates the explicit constraint on the problem~\eqref{eq:KL_softmax} in the local analysis regime.
\begin{lemma}\label{lem:sv:local}
If $X,Y$ are $\eps$-dependent, then the optimal $\weight, d$ for~\eqref{eq:KL_softmax} satisfy 
\begin{align}\label{eq:softmax_local}
|\weightp^\T(y)\rep(x) + \biasdp(y)| = O( \epsilon ), \quad \text{for all $x \in \cX, \ y \in \cY$}.
\end{align}
\end{lemma}
\begin{proof}
See Appendix~\ref{app:2}.
\end{proof}

In turn, we take~\eqref{eq:softmax_local} as the constraint for solving the problem~\eqref{eq:KL_softmax} in the local analysis regime. Moreover, we define the information vectors for zero-mean vectors $\repp$, $\weightp$ as $\rinf(x) = \sqrt{P_X(x)}\, \repp(x)$, $ \winf(y) = \sqrt{P_Y(y)}\, \weightp(y)$, and define matrices
\begin{align*}
\wmat \defeq \begin{bmatrix} \winf(1) & \cdots & \winf(|\Y|)
\end{bmatrix}^\T,\\
\rmat \defeq \begin{bmatrix} \rinf(1) & \cdots & \rinf(|\X|)
\end{bmatrix}^\T.
\end{align*}
\begin{lemma} \label{lem:local:kl}
The K-L divergence~\eqref{eq:KL_softmax} in the local analysis regime~\eqref{eq:softmax_local} can be expressed as 
\begin{align} \notag
&D( P_{XY} \| P_X \, \Pt^{(\weight,\bias)}_{Y|X}) \\ \label{eq:kl:local}
&= \frac{1}{2} \bfrob{ \dtm - \wmat \bigl(\rmat\bigr)^\T }^2 + \frac{1}{2} \lossb^{ (\weight,\bias) }(s) + o(\eps^2),
\end{align}
where $\lossb^{ (\weight,\bias) }(\rep)  \defeq \Ed{P_Y}{(\repm^\T \weightp(Y) + \biasdp(Y))^2}$.
\end{lemma}
\begin{proof}
See Appendix~\ref{app:3}.
\end{proof}
Eq.~\eqref{eq:kl:local} reveals key insights for feature selection in neural networks, which are illustrated by the following three learning problems, depending on if the weights, input feature, or both can be trained from data.



\subsection{Forward Feature Projection} \label{sec:FFP}


For the case that $\rep$ is fixed, we can optimize \eqref{eq:kl:local} with $\rmat$ fixed and get the following optimal weights: 
\begin{theorem}
For fixed $\rmat$ and $\repm$, the optimal $\wmat{}^*$ to minimize~\eqref{eq:kl:local} is given by
\begin{align} \label{eq:softmax_ace_f:1}
\wmat{}^* = \dtm\, \rmat\bigl(\bigl(\rmat\bigr)^\T \rmat\bigr)^{-1},
\end{align}
and the optimal weights $\weightp^*$ and bias $\biasdp^*$ are
\begin{align} \label{eq:softmax_ace_f:2}
\weightp^*(y) =  \Ed{P_{X|Y}}{\bLa_{\repp(X)}^{-1}\, \repp(X) \Bigm|
    Y=y}, \ \biasdp^*(y) = -\repm^\T\, \weightp(Y).
\end{align}
where $\bLa_{\repp(X)}$ denotes the covariance matrix of $\repp(X)$.
\label{thm:softmax:ace:f}
\end{theorem}

\begin{proof}
See Appendix~\ref{app:4}.
\end{proof}

Eq.~\eqref{eq:softmax_ace_f:1} can be viewed as a projection of the input feature $\repp(x)$, to a feature $\weight(y)$ computable from the value of $y$, which is the most correlated feature to $\repp(x)$. The solution is given by left multiplying the $\dtm$ matrix. We call this the ``forward feature projection''. 



\begin{remark} \label{remark:1}
While we assume the continuous input $\rep(x)$ is a function of a discrete variable $X$, we only need the labeled samples between $s$ and $Y$ to compute the weights and bias from the conditional expectation~\eqref{eq:softmax_ace_f:2}, and the correlation between $X$ and $s$ is irrelevant. Thus, our analysis for weights and bias can be applied to continuous input networks by just ignoring $X$ and taking $s$ as the real network input.
\end{remark}


\subsection{Backward Feature Projection}


It is also useful to consider the ``backward problem", which attempts to find informative feature $\rep^*(X)$ to minimize the
loss~\eqref{eq:kl:local} with given weights and bias.
\begin{theorem}
For fixed $\weightp$, $\wmat$, and $\biasdp$, the optimal $\rmat{}^*$ to minimize~\eqref{eq:kl:local} is given by 
\begin{align} \label{eq:softmax_ace_b:1}
\rmat{}^* = \dtm^\T \, \wmat\bigl(\bigr(\wmat\bigr)^\T
  \wmat\bigr)^{-1}, 
\end{align}
and the optimal feature function $\rep^*$, which are decomposed to $\repp^*$ and $\repm^*$, are given by
\begin{align} \notag
&\repp^*(x) =  \Ed{P_{Y|X}}{\bLa_{\weightp(Y)}^{-1}\, \weightp(Y)\middle|X=x}, \\ \label{eq:softmax_ace_b:2}
&\repm^* = -\bLa_{\weightp(Y)}^{-1}\, \Ed{P_Y}{\weightp(Y)\,\biasdp(Y)},
\end{align}
where $\bLa_{\weightp(Y)}$ denotes the covariance matrix of $\weightp(Y)$.
\label{thm:softmax:ace:b}
\end{theorem}

\begin{proof}
See Appendix~\ref{app:4}.
\end{proof}

The solution of this backward feature projection is precisely symmetric to the forward one. Note we assumed here that the feature $\rep(X)$ can be selected as any desired function. This is only true in the ideal case where the previous hidden layers of the neural network have sufficient expressive power. That is, it can generate the desired feature function as given in \eqref{eq:softmax_ace_b:2}. In general, however, the form of feature functions that can be generalized is often limited by the network structure. In the next section, we discuss such cases, where we do know the most desirable feature function as given in \eqref{eq:softmax_ace_b:2}, and the question is how does a network with limited expressive power approximate this optimal solution. 

\subsection{Universal Feature Selection}

When both $\rep$ and $(\weight, \bias)$ (and hence $\rmat, \wmat, d$) can be designed, the optimal $(\wmat, \rmat)$ corresponds to the low rank factorization of $\dtm$, and the solutions coincide with the universal feature selection.
\begin{theorem} \label{thm:softmax:ace}
The optimal solutions for weights and bias to maximize~\eqref{eq:kl:local} are given by $\biasdp(y) = -\repm^\T \weightp(y)$, and $(\wmat, \rmat)^*$ chosen as the largest $k$ left and right singular vectors of $\dtm$.
\end{theorem}

\begin{proof}
See Appendix~\ref{app:5}.
\end{proof}

Therefore, we conclude that the softmax regression, when both $\rep$ and $(\weight, \bias)$ are designable, is to extract the most correlated aspects of the input data $X$ and the label $Y$ that are informative features for data inferences from universal feature selection.

In the learning process of DNN, the BackProp procedure alternatively chooses the weights of the softmax layer and those on the previous layer(s). In each step, the weights on the rest of the network are fixed. This is equivalent as alternating between the forward and the backward feature projections, i.e. it alternates between \eqref{eq:softmax_ace_f:1} and \eqref{eq:softmax_ace_b:1}. This is in fact the power method to solve the SVD for $\dtm$~\cite{stoer2013introduction}, which is also known as the Alternating Conditional Expectation (ACE) algorithm~\cite{Breiman85}. 




\newcommand{\actfun}{\sigma} 
\newcommand{\bhref}{q} 
\newcommand{\bhjmat}{\mathbf{J}} 
\newcommand{\bhwmatv}{\mathbf{W}} 
\section{Multi-Layer Network Analysis}


From the previous discussions, the performance of the softmax
regression not only depends on the weight and bias $(\hweight(y),
\hbias(y))$, but the input feature $\hrep(x)$ has to be
informative. It turns out that the hidden layers of neural networks,
which are known to have strong expressive power of features, are
essentially extracting such informative features. For illustration, we
consider the neural network with a hidden layer of
$\hrdim$ nodes, and a zero-mean continuous input $\bhrep = [\bhrep_1 \ \cdots \ \bhrep_m]^{\T} \in \mathbb{R}^{\bhrdim}$ to this hidden layer, where $\bhrep$ is assumed to be a function $\bhrep(x)$ of some discrete variable $X$\footnote{As discussed in Remark~\ref{remark:1}, $X$ is assumed only for the convenience of analysis, and the computation of weights and bias only needs $\bhrep$, but not $X$. Moreover, the input $\bhrep$ to the hidden layer can be either directly from data or the output of previous hidden layers in a DNN, which we model as ``pre-processing" as shown in Fig. \ref{fig:MLP}.}. Our goal is to analyze the weights and bias in this layer with labeled samples $(\bhrep(x_i), y_i)$.
Assume the activation function of the hidden layer is a generally smooth function $\actfun(\cdot)$, then the output $\hrep_{z}(X)$ of the $z$-th hidden node is
\begin{align} \label{eq:layer_s}
\hrep_{z}(x) = \actfun\left(\bhweight^\T(z) \bhrep(x) + \bhbias(z)\right), \quad \text{for $z = 1, \ldots , k$}, \ x \in \cX,
\end{align}
where $\bhweight(z) \in \mathbb{R}^m$ and $\bhbias(z) \in \mathbb{R}$ are the weights and bias from input layer to hidden layer as shown in Fig. \ref{fig:MLP}. We denote $\hrep = [\hrep_{1} \ \cdots \ \hrep_{\hrdim}]^\T$ as the input vector to the output softmax regression layer.


\begin{figure}[t]
  \centering
  \resizebox {.8\columnwidth} {!}{\def\layersep{2cm}
\begin{tikzpicture}[shorten >=1pt,->, node distance=\layersep, draw=white!60!black!20!blue]
  \tikzstyle{every pin edge}=[<-,shorten <=1pt]
  \tikzstyle{neuron}=[line width = .8, draw = black, circle, minimum size=18pt,inner sep=0pt];
  \tikzstyle{cdot}=[draw=none, fill=none, execute at begin node=\color{black}$\cdots$];
  \tikzstyle{input neuron}=[neuron]; 
  \tikzstyle{hidden neuron}=[neuron]; 
  \tikzstyle{output neuron}=[neuron]; 
  \tikzstyle{annot} = [text width=4em, text centered];

  \node [input neuron] (input) at (0,-2.5 cm) {$X$};

  \node[rectangle, draw = black, line width = .8, fill = black!10!white, text width = 1.3em, minimum height = 3em] (rect) at (.5*\layersep,-2.5) {};
  \node[annot, text width = 1.3em] at (.5*\layersep,-2.5) {\tiny{\sf{Pre-Proc.}}};

  
\foreach \m [count=\y] in {1, missing, 2}
  \node [hidden neuron/.try, neuron \m/.try, yshift=-1.1cm] (hiddenA-\m) at (\layersep, -0.7 * \y cm) {};

\foreach \m [count=\y] in {$\bhrep_{1}$, , $\bhrep_{\bhrdim}$}
  \node [yshift=-1.1cm] at (\layersep, -0.7 * \y cm) {\small\m};

  
\foreach \m [count=\y] in {1, missing, 2}
  \node [hidden neuron/.try, neuron \m/.try, yshift=-1.1cm] (hiddenB-\m) at (2*\layersep, -0.7 * \y cm) {};

\foreach \m [count=\y] in {$\hrep_{1}$, , $\hrep_{\hrdim}$}
  \node [yshift=-1.1cm] at (2*\layersep, -0.7 * \y cm) {\small\m};

  \foreach \text / \m [count=\y] in {$Y = 1$/1, $Y = 2$/2, /missing, $Y = {\ydim}$/3}
  \node [rect neuron/.try, neuron \m/.try] (output-\m) at (3*\layersep,-\y) {\text}; 

  \node [bias neuron] (hiddenA-bias) at (\layersep, -4 cm) {\scriptsize{+1}};
  \node [bias neuron] (hiddenB-bias) at (2*\layersep, -4 cm) {\scriptsize{+1}};
    
\foreach \text [count=\i] in {1, 2}
   \draw [->, line width = .6, opacity = 0.4, draw=white!60!black!20!blue] (hiddenA-bias) -- (hiddenB-\i);
\foreach \text / \xs / \ys [count=\i] in {$\bhbias(1)$ / .0cm / -.1cm, $\bhbias(\hrdim)$  / .0cm / -.2cm}
   \draw [->, line width = .6, draw = none, opacity = 0.6] (hiddenA-bias) -- (hiddenB-\i) node[pos=0.25, xshift = \xs, yshift = \ys]{\tiny\text};

\foreach \text [count=\i] in {1, 2, 3}
   \draw [->, line width = .6, opacity = 0.4] (hiddenB-bias) -- (output-\i.west);
\foreach \text  / \xs / \ys [count=\i] in {$\hbias(1)$  / .0cm / .15cm, $\hbias(2)$ / .0cm / -.2cm, $\hbias(\ydim)$  / .0cm / -.11cm}
   \draw [->, line width = .6, draw = none, opacity = 0.6] (hiddenB-bias) -- (output-\i) node[pos=0.25, xshift = \xs, yshift = \ys]{\tiny\text};
   

\foreach \i in {1,...,2}
  \foreach \j in {1,...,2}
    \draw[->, line width = .6, draw=white!60!black!20!blue] (hiddenA-\i) -- (hiddenB-\j);

    \foreach \text / \xs / \ys [count=\i] in {$\bhweight(1)$ / 0cm / 0cm , $\bhweight(\hrdim)$ / 0cm / -0.1cm}
    \draw [draw = none] (hiddenA-missing) -- (hiddenB-\i) node[xshift = \xs, yshift = \ys, pos=0.75]{\scriptsize\text};
    
\foreach \i in {1,...,2}
  \foreach \j in {1,...,3}
    \draw [->, line width = .6] (hiddenB-\i) -- (output-\j.west);

\foreach \i in {1,...,2}
    \draw [->, line width = .6] (rect) -- (hiddenA-\i);
    \draw [->, line width = .6] (input) -- (rect);
    
\foreach \text / \xs / \ys [count=\i] in {$\hweight(1)$ / 0 cm / 0.2 cm, $\hweight(2)$ / 0 cm / 0.08 cm, $\hweight(\ydim)$ / 0 cm / -0.16 cm}
   \draw [draw = none] (hiddenB-missing) -- (output-\i) node[pos=0.65, xshift = \xs, yshift = \ys]{\scriptsize\text};
    


\end{tikzpicture}

}
 \caption{A multi-layer network: all hidden layers previous to $t$ are labeled as ``pre-processing".}
  \label{fig:MLP}
\end{figure}
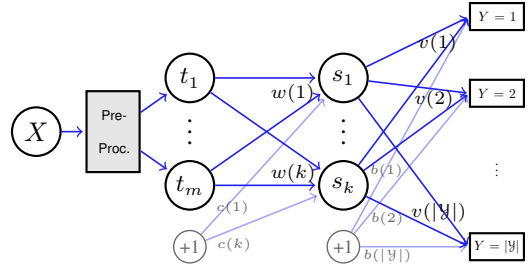

To interpret the feature selection in hidden layers, we fix $(\hweight(y), \hbias(y))$ at the output layer, and consider the problem of designing $(\bhweight(z), \bhbias(z))$ to minimize the loss function~\eqref{eq:KL_softmax} of the softmax regression at the output layer. Ideally, we should have picked $\bhweight(z)$ and $\bhbias(z)$ to generate $\hrep(x)$ to match $\hrep^*(x)$ from~\eqref{eq:softmax_ace_b:2}, which minimizes the loss. However, here we have the constraint that $\hrep(x)$ must take the form of~\eqref{eq:layer_s}, and intuitively the network should select $w(z), \bhbias(z)$ so that $\hrep(x)$ is close to $\hrep^*(x)$. Our goal is to quantify the notion of closeness in the local analysis regime.


To develop insights on feature selection in hidden layers, we again focus on the local analysis regime, where the weights and bias are assumed with the local constraint 
\begin{align} \label{eq:MLC}
\bigl|\weightp^\T(y)\rep(x) + \biasdp(y)\bigr| = O(\eps),   \  \bigl|\bhweight^\T(z)\bhrepp(x)\bigr|  = O(\eps), \ \forall x,y,z.
\end{align}
Then, since $\bhrep$ is zero-mean, we can express~\eqref{eq:layer_s} as
\begin{align} 
\hrep_{z}(x) 
&= \actfun\left(\bhweight^\T(z) \bhrep(x) + \bhbias(z)\right) \notag\\
&= \bhweight^\T(z)\bhrepp(x) \cdot \actfun'\left(\bhbias(z)\right) + \actfun\left(\bhbias(z)\right) + o(\eps),\label{eq:layer_s:2}
\end{align}
Moreover, we define a matrix $\hdtm$ with the $(z,x)$th entry $\hdtm(z,x) = \frac{\sqrt{P_X(x)}}{\actfun'(\bhbias(z))} \hrepp_{z}^{*}(x)$,
which can be interpreted as a generalized DTM for the hidden layer. Furthermore, we denote $\rinf_1(x) = \sqrt{P_X(x)}\, \bhrepp(x)$ as the information vector of $\bhrepp(x)$ with the matrix $\bhrmat$ defined as $\bhrmat \defeq \begin{bmatrix} \rinf_1(1) & \cdots & \rinf_1(|\X|) \end{bmatrix}^\T$, and we also define 
\begin{align*}
\bhwmatv &\defeq 
           \begin{bmatrix}
             \bhweight(1)& \cdots & \bhweight(\hrdim)
           \end{bmatrix}^{\T} \\
\bhjmat &\defeq \diag\{ \actfun'(\bhbias(1)), \actfun'(\bhbias(2)), \cdots, \actfun'(\bhbias(\hrdim)) \}.
\end{align*}
The following theorem characterizes the loss~\eqref{eq:KL_softmax}.

\begin{theorem} \label{thm:mlp:ace}
Given the weights and bias $(\weight, \bias)$ at the output layer, and for any input feature $\hrep$, we denote $\loss (\hrep)$ as the loss~\eqref{eq:KL_softmax} evaluated with respect to $(\weight, \bias)$ and $\hrep$. Then, with the constraints~\eqref{eq:MLC}
\begin{align} \notag
&\loss (\hrep) - \loss (\hrep^*) \\ \label{eq:thm2}
&= \frac{1}{2} \bfrob{ \prior \hdtm -  \prior\bhwmatv \bigl(\bhrmat\bigr)^\T}^2 + \frac{1}{2} \lossbh^{(\weight, \bias)} (\hrep, \hrep^*) + o(\eps^2),
\end{align}
where $\prior \defeq (\bigl(\wmat\bigr)^\T\wmat)^{1/2} \bhjmat$, and the term $\lossbh^{(\weight, \bias)} (\hrep, \hrep^*) $ $= (\hrepm - \hrepsm)^\T\bLa_{\weightp(Y)}(\hrepm - \hrepsm)$.
\end{theorem}
\begin{proof}
See Appendix~\ref{app:7}.
\end{proof}
Eq.~\eqref{eq:thm2} quantifies the closeness between $\hrep$ and $\hrep^*$ in terms of the loss~\eqref{eq:KL_softmax}. Then, our goal is to minimize~\eqref{eq:thm2}, which can be separated to two optimization problems:
\begin{align} \label{eq:hidden_UFS}
\bhwmatv^* &= \argmin_{\bhwmatv}\, \bbfrob{ \prior \hdtm - \prior \bhwmatv \bigl(\bhrmat\bigr)^\T}^2, \\ \label{eq:hidden_bias}
\repm^* &= \argmin_{\repm}\, \lossbh^{(\weight, \bias)} (\hrep, \hrep^*).
\end{align}
First note that the optimization problem~\eqref{eq:hidden_UFS} is similar to the ordinary softmax regression depicted in \secref{sec:softmax}, and the optimal solution is given by $\bhwmatv^* = \hdtm \bhrmat\bigl(\bigl(\bhrmat\bigr)^\T\bhrmat\bigr)^{-1}$. Therefore, solving the optimal weights in the hidden layer can be interpreted as projecting $\hrepp^*(x)$ to the subspace of feature functions spanned by $\bhrep(x)$ to find the closest expressible function. Finally, the problem~\eqref{eq:hidden_bias} is to choose $\hrepm$ (and hence the bias $\bhbias(z)$) to minimize the quadratic term similar to $\lossb^{(v,b)}(s)$ in~\eqref{eq:kl:local}, and we refer to \appref{app:7} 
for the optimal solution of~\eqref{eq:hidden_bias}.

Overall, we observe the correspondence between \eqref{eq:softmax_ace_f:1}, \eqref{eq:softmax_ace_b:2}, and \eqref{eq:hidden_UFS}, \eqref{eq:hidden_bias}, and interpret both operations as feature projections. Our argument can be generalized to any intermediate layer in a multi-layer network, with all the previous layers viewed as the fixed pre-processing that specifies $\bhrep(x)$, and all the layers after determining $\hrep^{*}$. Then the iterative procedure in back-propagation can be viewed as alternating projection finding the fixed-point solution over the entire network. This final fixed-point solution, even under the local assumption, might not be the SVD solution as in Theorem~\ref{thm:softmax:ace}. This is because the limited expressive power of the network often makes it impossible to generate the desired feature function. In such cases, the concept of feature projection can be used to quantify this gap, and thus to measure the quality of the selected features.

\section{Scoring Neural Networks}

Given a learning problem, it is useful to tell wether or not some extracted features is informative~\cite{alain2016understanding}. Our previous development naturally gives rise to a performance metric.
\begin{definition}
Given a feature $\rep(x)\in \mathbb{R}^k$ and weight $\weight(y)\in \mathbb{R}^k$ with the corresponding information matrices $\rmat$ and $\wmat$, the H-score $H(\rep, \weight)$ is defined as
\begin{align} \notag
H(\rep, \weight) 
&\defeq \frac{1}{2} \bfrob{\dtm}^2 -  \frac{1}{2} \bfrob{ \dtm - \wmat \bigl(\rmat\bigr)^\T }^2 \\ \label{eq:Hscore_sv1}
&=\Ed{P_{XY}}{\repp^{\T}(X)\, \weightp(Y)} - \frac{1}{2} \tr
\bigl( \bLa_{\repp(X)} \bLa_{\weightp(Y)} \bigr).
\end{align}
In addition, we define the single-sided H-score $H(\rep)$ as
\begin{align} \notag
H(\rep) 
&\defeq \frac{1}{2} \bfrob{ \dtm \rmat \bigl(\bigl(\rmat\bigr)^\T
  \rmat\bigr)^{-\frac{1}{2}} }^2  \\ \label{eq:Hscore_s1} 
  &= \frac{1}{2} \Ed{P_Y}{ \left\|
  \Ed{P_{X|Y}}{\bLa_{\repp(X)}^{-1/2} \, \repp(X) \Bigm| Y} \right\|^2 }.
\end{align}
\end{definition}
H-score can be used to measure the quality of features generated at any intermediate layer of the network. It is related to~\eqref{eq:thm2} when choosing the optimal bias and $\prior$ as the identity matrix.  
This can be understood as taking the output of this layer $\rep(x)$ and directly feed it to a softmax output layer with $\weight(y)$ as the weights, and $H(\rep, \weight)$ measures the resulting performance. Note that $\weight(y)$ here can be an arbitrary function of $Y$. It is not necessarily the weights on the next layer computed by the network. When the optimal weights $\weight^*(y)$ is used, the resulting performance becomes the one-sided H-score $H(\rep)$, which measures the quality of $\rep(x)$, and coincides with the information metric~\eqref{eq:EEg-k}.


In current practice the cross-entropy $\bE{\log \Pt^{(\weight, \bias)}_{Y|X}}$,
is often used as the performance metric. One can in principle also use
log-loss to measure the effectiveness of the selected feature at the
output of an intermediate layer~\cite{alain2016understanding}. However, one problem of this metric is that
for a given problem it is not clear what value of log-loss one should
expect because the log-loss is generally unbounded. Moreover, the computation of the log-loss for optimal weights and bias with respect to a particular input feature requires solving a non-convex optimization problem with the issue of locking at local optimum. 

In contrast, the H-score can be directly computed from the data samples, and has a clear upper bound from Lemma~\ref{lem:dtm:svd} that $H(s,v) \leq H(s) \leq
(1/2) \sum_{i=1}^k \sigma_i^2 \leq k/2$.  
In this sequence of inequalities, the gap over the first ``$\leq$" measures the optimality of the weights $\weight$; the second gap is due to the difference between the chosen feature and the optimal solution, which is a useful measure of how restrictive (lack of expressive power) the network structure is; and the last one measures how good the dataset itself is. In \secref{sec:exp}, we validate this metric in real data.


\section{Experimental Validation} \label{sec:exp}





We first validate the feature projection in Theorem~\ref{thm:softmax:ace}. For this purpose, we construct the NN as shown in \figref{fig:NN} with $\rdim = 1$, $|\mathcal{X}| = 8$, and $|\mathcal{Y}| = 6$, and the input feature $\rep(X)$ is generated from a sigmoid layer with the one-hot encoded $X$ as the input. Note that with proper weights in the sigmoid layer, $\rep(X)$ can express any desired function, up to scaling and shifting.
To compare the result trained by the neural network and that in \thmref{thm:softmax:ace}, we first randomly generate a distribution $P_{XY}$, and then generate $n = 100,000$ samples of $(X, Y)$ pair. Using these data to train the neural network, the corresponding results of $\rep(x), \weight(y)$ and $\bias(y)$ are shown in \figref{fig:exp:softmax-regression} with a comparison to theoretical result, where the training results match our theory. In addition, we validate \thmref{thm:mlp:ace} by the NN depicted in \figref{fig:MLP}, with the same setup of $X,Y$. The number of neurons in hidden layers are $\bhrdim = 4$ and $\hrdim = 3$, and the input $\bhrep(X)$ is some randomly chosen 
function of $X$, and the activation $\actfun(\cdot)$ is the sigmoid function. We then fix the weights and bias at the output layer and train the weights $\bhweight(1), \bhweight(2)$, $\bhweight(3)$, and bias $\bhbias$ in the hidden layer to optimize the Log-Loss.
\figref{fig:exp:softmax-regression} shows the matching between our results and the experiment.

\begin{figure}[t]
  \centering
    \begin{subfigure}{.3\columnwidth}
      \includegraphics[width = \columnwidth]{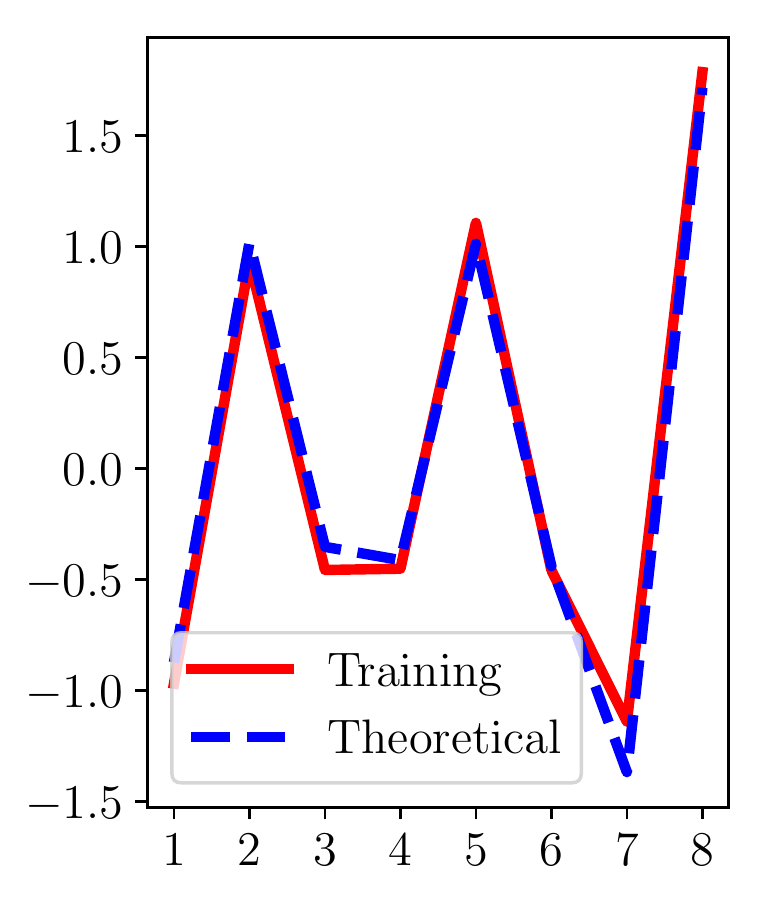}
    \caption{$\rep(x)$}
    \label{fig:exp:softmax-regression:s}
  \end{subfigure}
  \begin{subfigure}{.3\columnwidth}
    \includegraphics[width = \columnwidth]{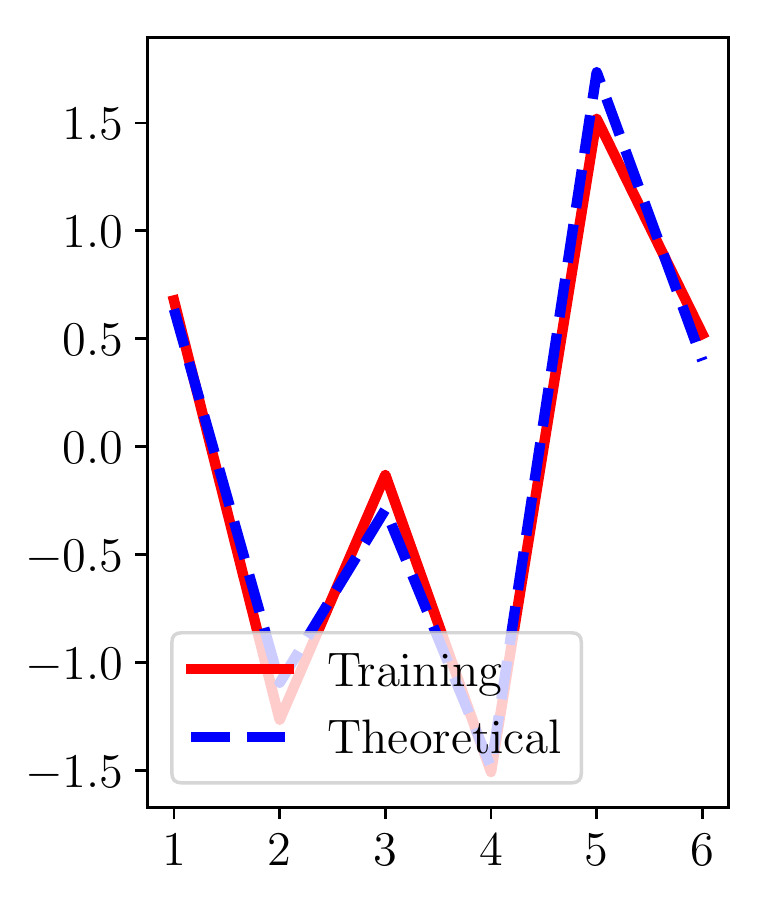}
    \caption{$\weight(y)$}
    \label{fig:exp:softmax-regression:v}
  \end{subfigure}
  \begin{subfigure}{.3\columnwidth}
    \includegraphics[width = \columnwidth]{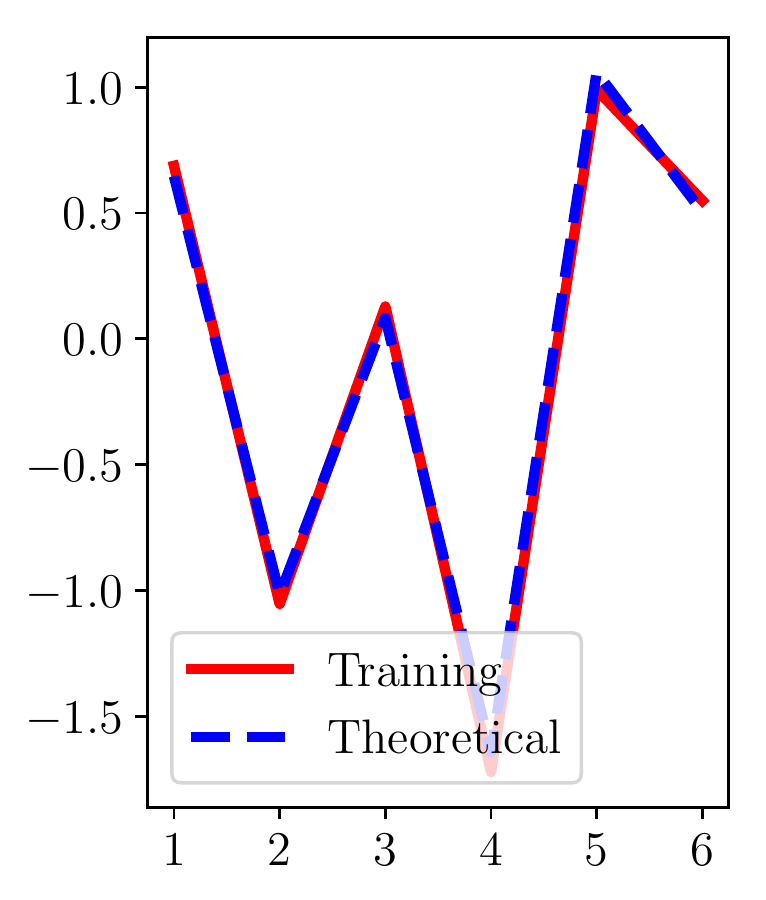}
    \caption{$\bias(y)$}
    \label{fig:exp:softmax-regression:b}
  \end{subfigure}
  \caption{The comparisons of the weights and bias in softmax regression.
  }
  \label{fig:exp:softmax-regression}
\end{figure}

\begin{figure}[t]
  \centering
  \hspace{-2em}
  \begin{subfigure}{.74\columnwidth}
    \includegraphics[height = 2.9cm]{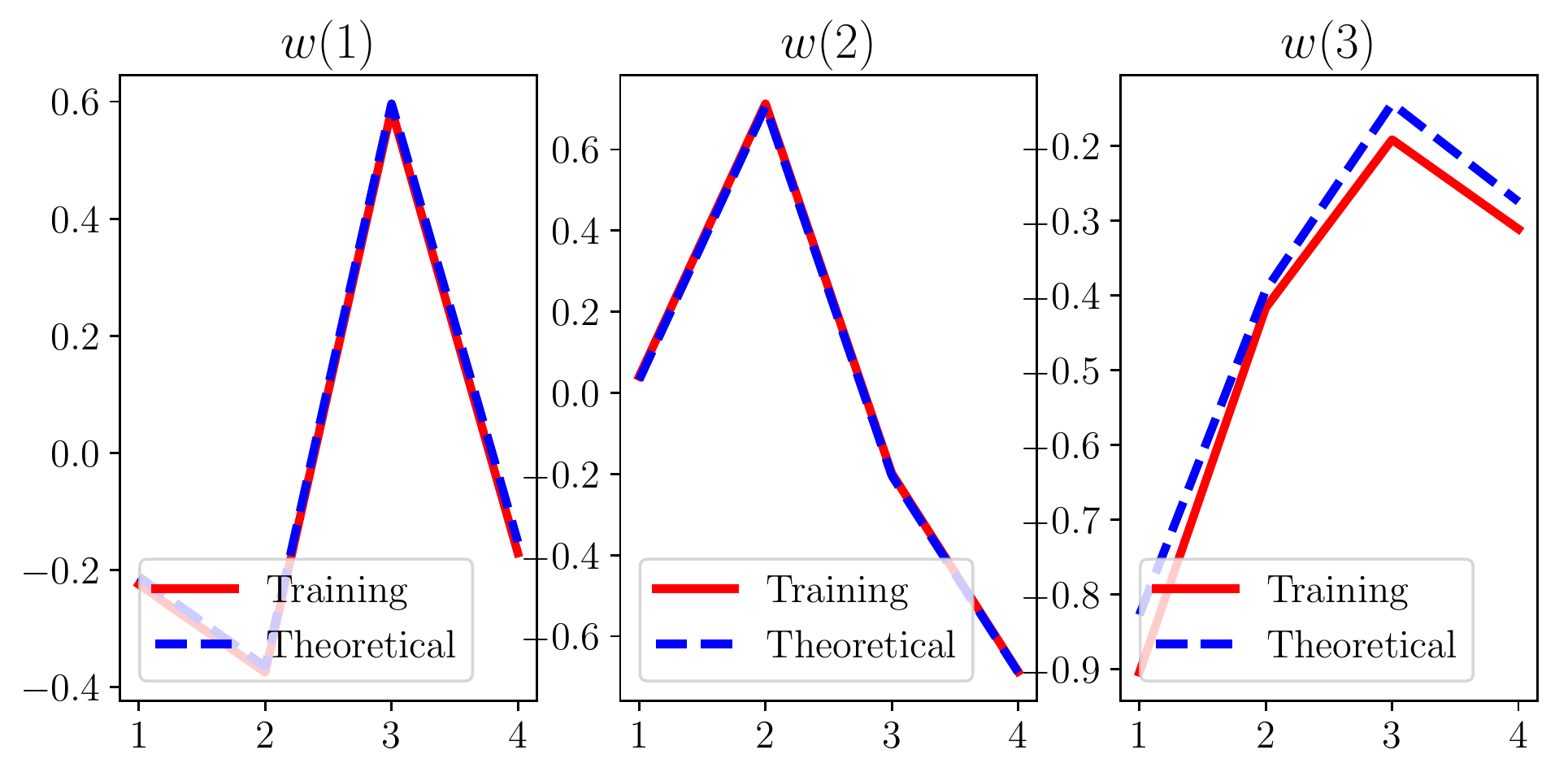}
    \caption{$\bhweight(1), \bhweight(2)$ and $\bhweight(3)$}
  \end{subfigure}
  \hspace{-1em}
  \begin{subfigure}{.24\columnwidth}
    \includegraphics[height = 2.9cm]{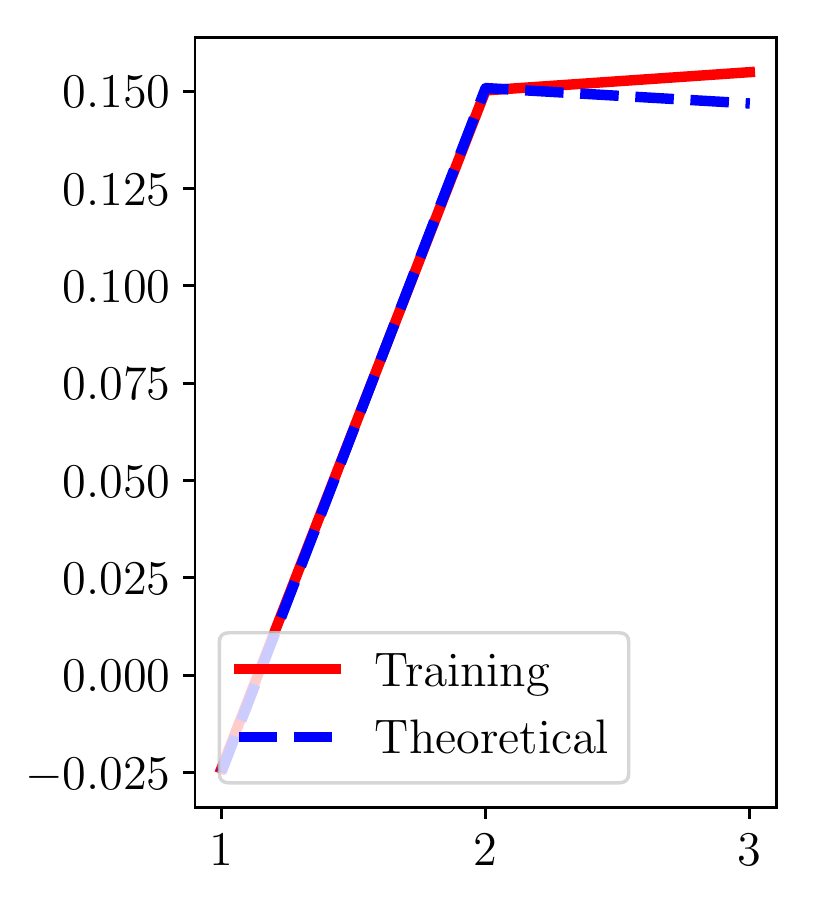}
    \caption{$\bhbias$}
  \end{subfigure}
  \caption{The comparisons of the weights and bias in the hidden layer.
  }
  \label{fig:exp:weight-hidden}
\end{figure}

\begin{table}[t]
  \centering
  \begin{tabular}{cccc}
    \toprule
    Model&          $H(\rep)$ &  $H_{\rm{AIC}}(\rep)$&     $\mathrm{Accaracy}$ \\
    \midrule
    VGG16&             148.3&    41.9& 0.642\\
    VGG19&             152.7&    42.2& 0.647\\
    MobileNet&          45.9&    42.6& 0.684\\
    DenseNet121&        59.5&    53.3& 0.714\\
    DenseNet169&       81.2&     70.2& 0.736\\
    DenseNet201&       89.1&     73.5& 0.744\\
    Xception&          179.8&    162.2& 0.775\\
    InceptionV3&       181.2&    162.9& 0.763\\
    InceptionResNetV2& 241.1&    198.1& 0.791\\
    \bottomrule
  \end{tabular}
  \captionof{table}[t]{$H_{\rm{AIC}}$ is the H-score with AIC correction.}
  \label{tab:h-score:imagenet}
\end{table}

Furthermore, we compare the performance of classification and the H-score evaluated from the features extracted from the last hidden layer of different DNNs. We use the ILSVRC2012~\cite{ILSVRC15} as our validation data set, and train several state-of-art DNNs~\cite{simonyan2014very, szegedy2017inception, chollet2016xception, szegedy2016rethinking, howard2017mobilenets, huang2017densely} to extract the features from the last hidden layers. The resulting H-scores are shown in TABLE \ref{tab:h-score:imagenet} with the comparison to the classification accuracy, where $H_{\mathrm{AIC}}$ is the H-score with the correction of Akaike information criterion (AIC) \cite{akaike1998information} to reduce overfitting. In particular, $H_{\mathrm{AIC}}(s)$ is given by
\begin{equation}
  H_{\mathrm{AIC}}(s) = H(s) - \frac{n_{\mathrm{p}}}{n_{\mathrm{s}}},
\end{equation}
where $n_{\mathrm{p}}$ is the number of parameters contained in the model, and $n_{\mathrm{s}} = 1,300,000$ is the number of training samples in ImageNet. The corrected H-score is consistent with the accuracy, which validates the H-score. 

\section*{Acknowledgment}
The research of Shao-Lun Huang was funded by the Natural Science Foundation of China 61807021,  Shenzhen Science and Technology Research and Development Funds (JCYJ20170818094022586), and Innovation and entrepreneurship project for overseas high-level talents of Shenzhen (KQJSCX2018032714403783).









\appendix
\newcommand{\svpd}{\tau} 
\newcommand{\TV}{1}
\renewcommand{\eqspace}{}

\subsection{Proof of Theorem~\ref{thm:UFS}} \label{app:1}
\newcommand{\xedit}[1]{\textcolor{red}{#1}}
We commence with the characterization of the error exponent.
\begin{lemma}
  Given a reference distribution $P_X \in \relint(\simpX)$, a constant $\eps > 0$ and integers $n$ and $k$, let $x_1, \cdots, x_n$ denote i.i.d. samples from one of $P_1$ or $P_2$, where $P_1, P_2 \in \nbhd_\eps^\X(P_X)$. To decide whether $P_1$ or $P_2$ is the generating distribution, a sequence of $k$-dimensional statistics $h^k = (h_1,\dots,h_k)$ is constructed as
  \begin{equation}
    h_i = \frac1n \sum_{l=1}^n f_i(x_l), \quad i =  1, \ldots , k,
    \label{eq:17}
  \end{equation}
  where $(f_1(X), \dots, f_k(X))$ are zero mean, unit-variance, and uncorrelated with respect to $P_X$, i.e.,
  \begin{subequations}
    \begin{gather}
      \Ed{P_X}{f_i(X)} = 0, \quad i \in \{1, \dots, k\}\label{eq:f:0mean}\\
      \Ed{P_X}{f_i(X)f_j(X)} = \delta_{ij}, \quad i, j \in \{1, \dots, k\}.\label{eq:f:uncorr}
    \end{gather}
    \label{eq:f:normal}
  \end{subequations}
  Then the error probability of the decision based on $h^k$ decays exponentially in $n$ as $n \to \infty$, with (Chernoff) exponent
  \begin{subequations}
  \begin{equation}
    \lim_{n \to \infty} \frac{-\log p_e}{n} \defeq E_{h^k} = \sum_{i = 1}^k E_{h_i},
    \label{eq:mm-exp-k}
  \end{equation}
  where
  \begin{equation}
    E_{h_i} = \frac{1}{8}\ip{\phi_1-\phi_2}{\fvgen_i}^2 + o(\eps^2),
  \end{equation}
  \label{eq:exponent-k}
\end{subequations}
   and $\phi_1 \leftrightarrow P_1, \phi_2 \leftrightarrow P_2, \fvgen_i \leftrightarrow f_i(X), i \in \{1, \dots, k\}$ are the corresponding information vectors.
\label{lemma:exponent:multi}  
\end{lemma}
\begin{proof}
  Since the rule is to decide based on comparing the projection
  \begin{equation*}
    \sum_{i = 1}^k h_i\bigl( \Ed{P_1}{f_i(X)} - \Ed{P_2}{f_i(X)}\bigr)
  \end{equation*}
  to a threshold, via Cram\'er's theorem \cite{dembo2010large} the error exponent under $P_j~(j = 1, 2)$ is
  \begin{equation}
    E_j(\la) = \min_{P \in \cS(\lambda)} D(P\|P_j),
    \label{eq:exp:opt}
  \end{equation}
  where
  \begin{align}
    &\cS(\la) \defeq \Bigl\{ P\in\simpX \colon \notag\\
    &\qquad\bEd{P}{f^k(X)} =  \la \, \bEd{P_1}{f^k(X)} + (1-\la) \, \bEd{P_2}{f^k(X)} \Bigr\}.
    \label{eq:cSk-def}
  \end{align}
  Now since \eqref{eq:f:0mean} holds, we obtain
  \begin{align} 
    \Ed{P_j}{f_i(X)} 
    &= \sum_{x\in\X} P_j(x)\, f_i(x) \notag\\
    &= \sum_{x\in\X} P_X(x)\, f_i(x) + \sum_{x\in\X}(P_j(x) - P_X(x))f_i(x)  \notag\\
    &= \Ed{P_X}{f_i(X)} + \sum_{x\in\X}  \sqrt{P_X(x)}\,\phi_j(x) \cdot  \frac{\fvgen_i(x)}{\sqrt{P_X(x)}} \notag\\ 
    &= \sum_{x\in\X} \phi_j(x)\,\fvgen_i(x) \notag\\
    &= \ip{\phi_j}{\fvgen_i}, \quad j=1, 2~\text{and}~ i =1, \ldots, k,
      \label{eq:h-mean}
  \end{align}
  which we express compactly as
  \begin{equation*}
    \bEd{P_j}{f^k(X)} = \ip{\phi_j}{\fvgen^k}, \quad j = 1, 2
  \end{equation*}
  with $\fvgen^k \defeq (\fvgen_1, \dots, \fvgen_k)$.

  Hence, the constraint \eqref{eq:cSk-def} is expressed in information vectors as
  \begin{equation*}
  \ip{\phi}{\fvgen_i} = \ip{\la \, \phi_1 + (1-\la) \, \phi_2}{\fvgen_i},
  \quad i = 1, \cdots, k,
\end{equation*}
i.e.,
\begin{equation}
  \bip{\phi}{\fvgen^k} = \bip{\la \, \phi_1 + (1-\la) \, \phi_2}{\fvgen^k}.
  \label{eq:linear:family}  
\end{equation}

  In turn, the optimal $P$ in \eqref{eq:exp:opt}, which we denoted by $P^*$, lies in the exponential family through $P_j$ with natural statistic $f^k(x)$, i.e., the $k$-dimensional family whose members are of the form
  \begin{equation*}
    \log \Pt_{\theta^k}(x) = \sum_{i = 1}^k \theta_i f_i(x) + \log P_j(x) - \alpha\bigl(\theta^k\bigr),
  \end{equation*}
  for which the associated information vector is
  \begin{equation}
    \phit_{\theta^k}(x) =  \sum_{i = 1}^k \theta_i \xi_i(x) + \phi_j(x) - \alpha(\theta^k) \sqrt{P_X(x)} + o(\eps),
    \label{eq:16}
  \end{equation}
  where we have used the fact that 
  \begin{align*}
    \log P(x) &= \log P_X(x) + \log \frac{P(x)}{P_X(x)}\\ 
    &= \log P_X(x) + \log \left(1 + \frac{1}{\sqrt{P_X(x)}}\phi(x)\right)\\
    &= \log P_X(x) + \frac{1}{\sqrt{P_X(x)}}\phi(x) + o(\eps)
  \end{align*}
  for all $P \in \nbhd_\eps^\X(P_X)$ with the information vector $\phi \leftrightarrow P$.
  As a result, 
  \begin{equation*}
    \langle\phit_{\theta^k}, \xi_i\rangle = \theta_i + \langle\phi_j, \xi_i\rangle + o(\eps),
  \end{equation*}
  where we have used \eqref{eq:f:uncorr}. Hence, via \eqref{eq:linear:family} we obtain that the intersection with the linear family \eqref{eq:cSk-def} is at $P^* = P_{{\theta^k}^*}$ with
  \begin{equation*}
    \theta_i^* = \langle\lambda\phi_1 + (1-\lambda)\phi_2 - \phi_j, \xi_i\rangle + o(\eps)
  \end{equation*}
  and thus
  \begin{subequations}
    \begin{align}
      E_j(\la)
      &= D(P^* \| P_j)\notag\\
      &= \frac12 \bigl\|\phit_{\theta^k} - \phi_j\bigr\|^2 + o(\eps^2) \label{eq:13}\\
      &= \frac12 \Biggl\| \sum_{i = 1}^k \theta_i^*\xi_i\Biggr\|^2 + \frac{1}{2}\alpha\bigl({\theta^{k}}^*\bigr)^2 + o(\eps^2)\label{eq:14}\\
      &= \frac12 \sum_{i = 1}^k (\theta_i^*)^2 + \frac12 \alpha\bigl({\theta^k}^*\bigr)^2 + o(\eps^2)\label{eq:15}\\
      &= \frac{1}{2} \sum_{i = 1}^k \ip{\lambda\phi_1 + (1 - \lambda)\phi_2 - \phi_j}{\xi_i}^2 + o(\eps^2),\label{eq:20}
    \end{align}
  \end{subequations}
  where to obtain \eqref{eq:13} we have exploited the local approximation of K-L divergence \cite{huang2017information}, to obtain \eqref{eq:14} we have exploited \eqref{eq:16}, to obtain \eqref{eq:15} we have again exploited \eqref{eq:f:uncorr}, and to obtain \eqref{eq:20} we have used that
  \begin{equation*}
    \alpha\bigl({\theta^k}^*\bigr) = o(\eps^2)
  \end{equation*}
  since ${\theta^{k}}^* = O(\eps)$ and
  \begin{equation*}
    \alpha(0) = 0, \quad\text{and}\quad \nabla\alpha(0) = \bEd{P_j}{f^k(X)} = \ip{\phi_j}{\xi^k} = O(\eps).
  \end{equation*}
  Finally, $E_1(\lambda) = E_2(\lambda)$ when $\lambda = 1/2$, so the overall error probability has exponent \eqref{eq:exponent-k}.
\end{proof}

Then, the following lemma demonstrates a property of information vectors in a Markov chain.
\begin{lemma}
  \label{lem:markov:iv}
  Given the Markov relation $X\leftrightarrow Y\leftrightarrow V$ and any $v \in \cV$, let $\bphi_v^{X|V}$ and $\bphi_v^{Y|V}$ denote the associated information vectors for $P_{X|V}(\cdot|v)$ and $P_{Y|V}(\cdot|v)$, then we have
  \begin{equation}
    \bphi^{X|V}_v = \dtm^\T{\bphi}^{Y|V}_v.
    \label{eq:infovec:x-y}
  \end{equation}

\end{lemma}
\begin{proof}
The Markov relation implies
\begin{equation*}
  P_X(x) = \sum_{y \in \Y}P_{X|Y}(x|y) P_{Y}(y),
\end{equation*}
\begin{align*}
  P_{X|V}(x|v) &= \sum_{y \in \Y}P_{X|Y,V}(x|y,v) P_{Y|V}(y|v)\\
  &= \sum_{y \in \Y}P_{X|Y}(x|y) P_{Y|V}(y|v).
\end{align*}
As a result,
\begin{equation*}
    P_{X|V}(x|v) - P_X(x) = \sum_{y \in \Y}P_{X|Y}(x|y) [P_{Y|V}(y|v) - P_Y(y)],
\end{equation*}
and the corresponding information vectors satisfy 
\begin{align}
  \phi^{X|V}_v(x)
  &= \frac{1}{\sqrt{P_X(x)}} \sum_{y \in \Y}P_{X|Y}(x|y)\sqrt{P_Y(y)}\phi^{Y|V}_v(y)\notag\\
  &= \sum_{y \in \Y} \left[\dtm (y,x) + \sqrt{P_X(x)P_Y(y)}\right]\phi^{Y|V}_v(y)\notag\\
  &= \sum_{y \in \Y} \dtm (y,x) \phi^{Y|V}_v(y)\label{eq:ortho},
\end{align}
where the last equality follows from the fact that
\begin{equation*}
  \sum_{y \in \Y}\sqrt{P_Y(y)}\phi^{Y|V}_v(y) = \sum_{y \in \Y} [P_{Y|V}(y|v) - P_Y(y)]= 0.
\end{equation*}
Finally, express \eqref{eq:ortho} in the matrix form and we obtain \eqref{eq:infovec:x-y}.
\end{proof}

In addition, the following lemma is useful for dealing with the expectation over an RIE.
\begin{lemma}
  Let $\bz$ be a spherically symmetric random vector of dimension $M$,
  i.e., for any orthogonal $\bQ$ we have $\bz\eqd \bQ\bz$. If $\bA$ is a fixed matrix of compatible dimensions, then
  \begin{equation}
    \E{\|\bz^\T \bA\|^2} = \frac{1}{M}\E{\|\bz\|^2}\frob{\bA}^2.
  \end{equation}  
  \label{lem:rie}    
\end{lemma}
\begin{proof}
  By definition we have $\bLa_\bz = \bQ\bLa_\bz \bQ^\T$ for any orthogonal $\bQ$, hence $\bLa_\bz$ is diagonal. Suppose $\bLa_\bz = \lambda\, \bI$, then from
  \begin{equation*}
    \tr\left(\bLa_\bz\right) = \E{\|\bz\|^2} = \lambda M
  \end{equation*}
  we obtain
  \begin{equation*}
    \lambda = \frac{1}{M} \tr\left(\bLa_\bz\right).
  \end{equation*}  
  As a result,
  \begin{equation}
     \begin{aligned}
      \E{\|\bz^\T\bA\|^2} &= \tr\left(\bA^\T\bLa_\bz\bA\right) = \lambda \tr\left(\bA^\T\bA\right)\\
      &= \frac{1}{M}\E{\|\bz\|^2}\frob{\bA}^2.
   \end{aligned}
  \end{equation}
\end{proof}

We now have everything to prove~\thmref{thm:UFS}.
\begin{proof}[Proof of \thmref{thm:UFS}]

By definition of feature functions, we have $\bEd{P_X}{f_i(X)} = 0, i = 1, \dots, k$. 
Suppose $\bm{f}$ is the vector representation of $f^k$ and denote by $\tilde{\bm{f}} \defeq \bLa_{\bm{f}}^{-1/2}\bm{f}$ the normalized $\bm{f}$, with $\bLa_{\bm{f}}^{1/2}$ denoting any square root matrix of $\bLa_{\bm{f}}$, then the corresponding statistics $\tilde{f}^k = (\tilde{f}_1, \dots, \tilde{f}_k)$ satisfy the constraints \eqref{eq:f:normal}. Further, construct the statistic $\tilde{h}^k = (\tilde{h}_1, \dots, \tilde{h}_k)$ as [cf. \eqref{eq:17}]
  \begin{equation}
    \tilde{h}_i = \frac1n \sum_{l=1}^n \tilde{f}_i(x_l), \quad i =  1, \ldots , k.
  \end{equation}
  Then, from \lemref{lemma:exponent:multi}, the error exponent of distinguishing $v$ and $v'$ based on $\tilde{h}^k$ is
\begin{align*}
  E_{\tilde{h}^k}(v, v')
  &= \frac{1}{8}\sum_{i = 1}^{k}\left[ \big({\bphi}^{X|V}_v - {\bphi}^{X|V}_{v'}\big)^\T \tilde{\bxi}^X_i \right]^2 + o(\eps^2)\\
  &= \frac{1}{8}\left\| \big({\bphi}^{X|V}_v - {\bphi}^{X|V}_{v'} \big)^\T \tilde{\bXi}^X\right\|^2 + o(\eps^2),
\end{align*}
where ${\bphi}_v^{X|V}$ denotes the associated information vector for $P_{X|V}(\cdot|v)$, $\tilde{\bxi}_i^X$ denotes the information vectors of $\tilde{f}_i$, and $\tilde{\bXi}^X \defeq [\tilde{\bxi}^X_1, \dots, \tilde{\bxi}^X_k]$. Since the optimal decision rule is linear, the error exponent is invariant with linear transformations of statistics, i.e.,
  \begin{align}
     E_{{h}^k}(v, v')
    &= E_{\tilde{h}^k}(v, v')\notag\\
    &= \frac{1}{8}\left\| \big({\bphi}^{X|V}_v - {\bphi}^{X|V}_{v'} \big)^\T \tilde{\bXi}^X\right\|^2 + o(\eps^2)\notag\\
    &= \frac{1}{8}\left\| \big({\bphi}^{Y|V}_v - {\bphi}^{Y|V}_{v'} \big)^\T \dtm\tilde{\bXi}^X\right\|^2 + o(\eps^2), \label{eq:18}
  \end{align}
  where the last equality follows from \lemref{lem:markov:iv}. Taking the expectation over a given RIE yields
  \begin{align*}
    \E{E_{{h}^k}(v, v')}
    &= \frac{1}{8}\E{\left\| \big({\bphi}^{Y|V}_v - {\bphi}^{Y|V}_{v'} \big)^\T \dtm\tilde{\bXi}^X\right\|^2} + o(\eps^2)\\
    &= \frac{\E{\big\|{\bphi}^{Y|V}_v - {\bphi}^{Y|V}_{v'}\big\|^2}}{8|\Y|}\bfrob{\dtm\tilde{\bXi}^X}^2 + o(\eps^2),
  \end{align*}
  where we have exploited \lemref{lem:rie}. Finally, the error exponent \eqref{eq:EEg-k} can be obtained via noting
\begin{equation*}
  \tilde{\bXi}^X = \bXi^X \bigl(\bigl(\rmat\bigr)^\T \rmat\bigr)^{-\frac{1}{2}}
\end{equation*}
from the definition of $\tilde{f}^k$.  

\end{proof}

\subsection{Proof of Lemma~\ref{lem:sv:local}} \label{app:2}
We first prove two useful lemmas.

\begin{lemma}
For distributions $P \in \relint(\simpX)$, $Q, R \in \simpX$, and sufficiently small $\eps$, if $D(P\|Q) \leq \eps^2$ and $D(P\|R) \leq \eps^2$, then there exists a constant $C > 0$ independent of $\eps$, such that  $D(Q\|R) \leq C\eps^2$.
  \label{lem:kl:continuous}
\end{lemma}
\begin{proof}
  Denote by $\|\cdot\|_1$ the $\ell_1$-distance between distributions, i.e., $\|P - Q\|_{\TV} \defeq \sum_{x \in \X} |P(x) - Q(x)|$, then from Pinsker's inequality \cite{cover2012elements}, we have
  \begin{subequations}
    \begin{align}
      \|P - Q\|_{\TV} &\leq \sqrt{2 D(P\|Q)} < \sqrt{2}\eps,\label{eq:pinsker:1}\\
      \|P - R\|_{\TV} &\leq \sqrt{2 D(P\|R)} < \sqrt{2}\eps,\label{eq:pinsker:2}
    \end{align}
  \end{subequations}
  which implies
  \begin{align}
    \|Q - R\|_{\TV} \leq \|P - Q\|_{\TV} + \|P - R\|_{\TV} \leq  2\sqrt{2}\eps.
    \label{eq:23}
  \end{align} 
  In addition, with the notation $p_{\min}\defeq \min_{x \in \cX} P(x)$, for all $x \in \X$ we have
  \begin{subequations}
  \begin{align}
    R(x) &> P(x) - |P(x) - R(x)|    \label{eq:19:1}\\
         &> \min_{x \in \cX} P(x) - \sqrt{2}\eps \label{eq:19:2}\\
         &= p_{\min} - \sqrt{2}\eps, \label{eq:19:3}
  \end{align}
\end{subequations}
where to obtain \eqref{eq:19:2} we have used \eqref{eq:pinsker:2}. Note that $p_{\min} > 0$ since $P \in \relint(\simpX)$, thus $R(x) > p_{\min}/2$ for sufficiently small $\eps$. As a result,
  \begin{subequations}
    \begin{align}
      D(Q\|R)
        &\leq \sum_{x \in \X} \frac{(Q(x) - R(x))^2}{R(x)}\label{eq:21:1}\\
        &\leq \frac{2}{p_{\min}} \sum_{x \in \X} [Q(x) - R(x)]^2\label{eq:21:2}\\
        &\leq \frac{2\|Q - R\|_{\TV}^2}{p_{\min}} \label{eq:21:3}\\
        &\leq \frac{16}{p_{\min}}\eps^2,\label{eq:21:5}
    \end{align}
  \end{subequations}
  where to obtain \eqref{eq:21:1} we have applied an upper bound of K-L divergence \cite{polyanskiy2014lecture}, 
  and to obtain \eqref{eq:21:5} we have used \eqref{eq:23}.
\end{proof}

\begin{lemma}
  For all $(x, y) \in \X\times \Y$, we have  
    \begin{align*}
    &\eqspace D(P_X P_Y \| P_X \, \Pt^{(\weight,\bias)}_{Y|X})\\
    &\geq P_X(x) \log \left[P_{Y}(y) e^{\svpd(x, y)} + (1 - P_Y(y))e^{-\frac{P_Y(y)}{1 - P_Y(y)}\svpd(x, y)}\right]
  \end{align*}
  where $\Pt^{(\weight,\bias)}_{Y|X}$ is defined in \eqref{eq:softmax} and $\svpd(x, y)$ is defined as $\svpd(x, y) \defeq \weightp^\T(y)\rep(x) + \biasdp(y)$.
  \label{lem:kl:pxpy-q:bound}
\end{lemma}
\begin{proof}
First, we can always 
use $(\weightp, \biasdp)$ to replace $(\weight, \bias)$, since
  \begin{align}
    \Pt^{(\weight,\bias)}_{Y|X}(y|x) &= \frac{e^{\weight^\T(y)\rep(x) + \bias(y)}}{\sum_{y' \in \Y}e^{\weight^\T(y')\rep(x) + \bias(y')}}\notag\\
    &= \frac{P_Y(y)e^{\weight^\T(y)\rep(x) + \biasd(y)}}{\sum_{y' \in \Y}P_Y(y)e^{\weight^\T(y')\rep(x) + \biasd(y')}}\notag\\
    &= \frac{P_Y(y)e^{\weightp^\T(y)\rep(x) + \biasdp(y)}}{\sum_{y' \in \Y}P_Y(y)e^{\weightp^\T(y')\rep(x) + \biasdp(y')}}\notag\\
    &= \frac{P_Y(y)e^{\svpd(x, y)}}{\sum_{y' \in \Y}P_Y(y)e^{\svpd(x, y')}}.\label{eq:11}
\end{align}

Then the K-L divergence $ D(P_X P_Y \| P_X \, \Pt^{(\weight,\bias)}_{Y|X})$ can be expressed as
\begin{align}
  &\eqspace D(P_X P_Y \| P_X \, \Pt^{(\weight,\bias)}_{Y|X})\notag\\    
  &= \sum_{(x, y) \in \X \times \Y} P_X(x)P_Y(y) \log \frac{\sum_{y' \in \Y}P_Y(y)e^{\svpd(x, y')}}{e^{\svpd(x, y)}}\notag\\
  &= \sum_{x \in \X} P_X(x) \log \Bigg[\sum_{y' \in \Y}P_Y(y)e^{\svpd(x, y')} \Bigg] - \Ed{P_XP_Y}{\svpd(X, Y)}\notag\\
  &= \sum_{x \in \X} P_X(x) \log \Bigg[\sum_{y' \in \Y}P_Y(y)e^{\svpd(x, y')} \Bigg],  \label{eq:kl:pxpy:q}
\end{align}
where to obtain the last equality we have used the fact $\Ed{P_XP_Y}{\svpd(X, Y)} = 0$. As a result, we have
\begin{subequations}
  \begin{align}
    &\eqspace D(P_X P_Y \| P_X \, \Pt^{(\weight,\bias)}_{Y|X})\\
    &\geq P_X(x) \log \Bigg[\sum_{y' \in \Y}P_Y(y')e^{\svpd(x, y')} \Bigg]\\
    &\geq P_X(x) \log \left[P_{Y}(y) e^{\svpd(x, y)} + (1 -
      P_Y(y))e^{-\frac{P_Y(y)}{1 - P_Y(y)}\svpd(x, y)}\right],\label{eq:kl:pxpy-q:bound:2}
  \end{align}
  \label{eq:kl:pxpy-q:bound}
\end{subequations}
where 
\eqref{eq:kl:pxpy-q:bound:2}
follows from 
Jensen's inequality:
   \begin{align*}
     &\eqspace\sum_{y' \in \Y}P_Y(y')e^{\svpd(x, y')}\\
     &= P_Y(y)e^{\svpd(x, y)} + (1 - P_Y(y))\sum_{y' \neq y}\frac{P_Y(y')}{1 - P_Y(y)}e^{\svpd(x, y')}\\
     &\geq P_Y(y)e^{\svpd(x, y)} \\
     &\qquad+ (1 - P_Y(y)) \exp\Bigg(\frac{1}{1 - P_Y(y)}\sum_{y' \neq y}P_Y(y')\svpd(x, y')\Bigg)\\
     &= P_{Y}(y) e^{\svpd(x, y)} + (1 -
       P_Y(y))e^{-\frac{P_Y(y)}{1 - P_Y(y)}\svpd(x, y)}.
   \end{align*}
\end{proof}

With the above lemmas, \lemref{lem:sv:local} can be proved as follows.
\begin{proof}[Proof of \lemref{lem:sv:local}]  
  Note that when $\weight = \biasd = 0$, we have $\Pt^{(\weight,\bias)}_{Y|X} = P_Y$. As a result, the optimal $\weight, \biasd$ for \eqref{eq:KL_softmax} satisfy 
  \begin{equation}
    \begin{aligned}
     &\eqspace D( P_{XY} \| P_X \, \Pt^{(\weight,\bias)}_{Y|X})\\
      &\leq D( P_{XY} \| P_X P_Y)\\
      &\leq \sum_{(x, y) \in \X \times \Y}\frac{\left[P_{X, Y}(x, y) - P_X(x)P_Y(y)\right]^2}{P_X(x)P_Y(y)}\\
      &\leq\eps^2,
\end{aligned}
\end{equation}
where the second inequality again follows from the upper bound for K-L divergence \cite{polyanskiy2014lecture}, and the last inequality follows from the definition of $\eps$-dependency.

As $P_{XY} \in \relint(\simpXY)$, from \lemref{lem:kl:continuous}, there exist $C > 0$ and $\eps_1 > 0$ such that  $D(P_X P_Y \| P_X \, \Pt^{(\weight,\bias)}_{Y|X}) < C \eps^2$ for all $\eps < \eps_1$. Furthermore, from \lemref{lem:kl:pxpy-q:bound}, for all $(x, y) \in \X \times \Y$ and $\eps \in (0, \eps_1)$, we have 
  \begin{align}
     C\eps^2
    &\geq P_X(x) \log \Bigl[P_{Y}(y) e^{\svpd(x, y)}
    + (1 - P_Y(y))e^{-\frac{P_Y(y)}{1 - P_Y(y)}\svpd(x, y)}\Bigr].
    \label{eq:4}
  \end{align}  
  Since
  \begin{align*}
    \log \Bigl[P_{Y}(y) e^{\svpd(x, y)}
       + (1 - P_Y(y))e^{-\frac{P_Y(y)}{1 - P_Y(y)}\svpd(x, y)}\Bigr].\notag\\
    = \frac{P_Y(y)}{2(1 - P_Y(y))} \svpd^2(x, y) + o(\svpd^2(x, y)),
  \end{align*}
  there exists $\delta > 0$ independent of $\eps_1$, such that for all $|\svpd(x, y)| \leq \delta$, we have
  \begin{equation}
    \begin{aligned}
      &\log \left[P_{Y}(y) e^{\svpd(x, y)} + (1 -  P_Y(y))e^{-\frac{P_Y(y)}{1 - P_Y(y)}\svpd(x, y)}\right]\\
      &\qquad
      > \frac{P_Y(y)}{2} \svpd^2(x, y).
    \end{aligned}
    \label{eq:1}
  \end{equation}

In turn, if $|\svpd(x, y)| > \delta$, we have
\begin{align*}
  &\eqspace \log \biggl[P_{Y}(y) e^{\svpd(x, y)} + (1 - P_Y(y))e^{-\frac{P_Y(y)}{1 - P_Y(y)}\svpd(x, y)}\biggr]\\
  &\geq \min\biggl\{\log \biggl[P_{Y}(y) e^\delta + (1 - P_Y(y))e^{-\frac{P_Y(y)}{1 - P_Y(y)}\delta}\biggr],\\
  &\qquad\qquad \log \biggl[P_{Y}(y) e^{-\delta} + (1 - P_Y(y))e^{\frac{P_Y(y)}{1 - P_Y(y)}\delta}\biggr]\biggr\},\\
  &\geq \frac{P_Y(y)}{2} \delta^2,
\end{align*}
where to obtain the second inequality we have exploited the monotonicity of function $P_{Y}(y) e^{t} + (1 - P_Y(y))e^{-\frac{P_Y(y)}{1 - P_Y(y)}t}$, and to obtain the third inequality we have exploited \eqref{eq:1}.

As a result, we have
\begin{align}
  &\log \left[P_{Y}(y) e^{\svpd(x, y)} + (1 -  P_Y(y))e^{-\frac{P_Y(y)}{1 - P_Y(y)}\svpd(x, y)}\right]\notag\\
  &\qquad> \frac{P_Y(y)}{2}\cdot \min\{\delta^2, \svpd^2(x, y)\},\label{eq:2}
\end{align}
hence \eqref{eq:4} becomes
\begin{align}
  C\eps^2 \geq \frac{P_X(x)P_Y(y)}{2}\cdot \min\{\delta^2, \svpd^2(x, y)\},
  \label{eq:5}
\end{align}
from which we obtain $\svpd(x, y) = O(\eps)$. Indeed, let $\eps_2 \defeq \frac{\delta}{\sqrt{2C}}\cdot\min_{(x, y)\in \cX \times \cY}\sqrt{P_X(x)P_Y(y)}$, $\eps_0 \defeq \min\{\eps_1, \eps_2\}$, then for all $\eps < \eps_0$, we have
\begin{equation*}
  C\eps^2 < \frac{P_X(x)P_Y(y)}{2}\cdot \delta^2,
\end{equation*}
and \eqref{eq:5} implies $|\svpd(x, y)| <  C'\eps$ with $ C' = \sqrt{\frac{2C}{P_X(x)P_Y(y)}}$.


\end{proof}

\subsection{Proof of Lemma~\ref{lem:local:kl}} \label{app:3}

\begin{proof}
  From \lemref{lem:sv:local}, there exists $C' > 0$ such that for all $(x, y) \in \X \times \Y$, we have
  \begin{equation}
    |\weightp^\T(y)\rep(x) + \biasdp(y)| < C'\eps,
    \label{local:1:1}
  \end{equation}
  which implies
  \begin{gather}
   |\repm^\T\weightp(y) + \biasdp(y)| < C\eps,\label{local:1:2}\\
   |\weightp^\T(y)\repp(x)| < 2C\eps, \label{local:1:3}
 \end{gather}
 with $C = \max\{C', 1\}$.

 From \eqref{eq:11}, we can assume $\Ed{P_{Y}}{\weight(Y)} = \Ed{P_{Y}}{\biasd(Y)} = 0$ without loss of generality. %
 Then \eqref{eq:softmax} can be rewritten as
\begin{equation}
  \begin{aligned}
    \Pt_{Y|X}^{(\weight,\bias)} (y|x)
    &= \frac{P_Y(y)e^{\weightp^\T(y)\rep(x) + \biasdp(y)}}{\sum_{y' \in \Y}P_Y(y')e^{\weightp^\T(y')\rep(x) + \biasdp(y')}},\\
  \end{aligned}
  \label{eq:softmax:new}
\end{equation}
and the numerator has the approximation
\begin{align*}
  &\eqspace
    P_Y(y)e^{\weightp^\T(y)\rep(x) + \biasdp(y)}\\
  &=P_Y(y)\left(1 + \weightp^\T(y)\rep(x) + \biasdp(y) + o(\eps)\right)\\
  &=P_Y(y)\left(1 + \weightp^\T(y)\rep(x) + \biasdp(y)\right) + o(\eps),
\end{align*}
where we have used \eqref{local:1:1}. Similarly, from
\begin{align*}
  &\eqspace
    \sum_{y' \in \Y}P_Y(y)e^{\weightp^\T(y)\rep(x) + \biasdp(y)}\\
  &=\sum_{y' \in \Y}P_Y(y)\left(1 + \weightp^\T(y)\rep(x) + \biasdp(y)\right) + o(\eps)\\
  &=1 + \Ed{P_Y}{\weightp^\T(Y)\rep(x)} + \Ed{P_Y}{\biasdp(y)} + o(\eps)\\    
  &=1 + o(\eps)
\end{align*}
we obtain
\begin{equation*}
\frac{1}{\sum_{y' \in \Y}P_Y(y)e^{\weightp^\T(y)\rep(x) + \biasdp(y)}} =\frac{1}{1 + o(\eps)} = 1 + o(\eps).
\end{equation*}
As a result, \eqref{eq:softmax:new} has the approximation
\begin{equation}
  \begin{aligned}
    &\eqspace
    \Pt_{Y|X}^{(\weight,\bias)} (y|x)\\
    &= \left[P_Y(y)\left(1 + \weightp^\T(y)\rep(x) + \biasdp(y)\right) + o(\eps)\right][1 + o(\eps)]\\
    &= P_Y(y)\left(1 + \weightp^\T(y)\rep(x) + \biasdp(y)\right) + o(\eps),\\
  \end{aligned} 
\end{equation}
which implies $P_X \,\Pt_{Y|X}^{(\weight,\bias)} \in \nbhd_{C\eps}^{\X\times \Y}(P_XP_Y)$ for sufficiently small $\eps$. Besides, the local assumption of distributions implies that $P_{X Y} \in \nbhd_\eps^{\X\times \Y}(P_XP_Y) \subset \nbhd_{C\eps}^{\X\times \Y}(P_XP_Y)$. Again, from the local approximation of K-L divergence \cite{huang2017information}
\begin{equation}
  D(P_1 \| P_2) = \frac{1}{2} \|\phi_1 - \phi_2\|^2 + o\big(\eps^2\big),
\end{equation}
we have
  \begin{align*}
&\eqspace D( P_{Y,X} \| P_X \, \Pt^{(\weight,\bias)}_{Y|X}) \\
&= \frac{1}{2} \sum_{x\in \X, y \in \Y} \frac{\left[P_{Y,X}(y, x) -  \Pt^{(\weight,\bias)}_{Y|X}(y|x) P_X(x)\right]^2}{P_{Y}(y)P_{X}(x)} + o(\eps^2)\\
&= \frac{1}{2} \sum_{x\in \X, y \in \Y} \left[\frac{P_{Y,X}(y, x)}{\sqrt{P_{Y}(y)P_{X}(x)}} - \sqrt{P_{Y}(y)P_{X}(x)}\right.\\
&\eqspace\quad\left. - \sqrt{P_{Y}(y)P_{X}(x)}\left(\weightp^\T(y)\rep(x) + \biasdp(y) + o(\eps)\right)\right]^2 + o(\eps^2)\\
&=  \frac{1}{2} \sum_{x\in \X, y \in \Y} \left[\dtm(y, x) - \sqrt{P_{Y}(y)P_{X}(x)}{\weightp^\T(y)\repp(x)}\right.\\
&\qquad \left.- \sqrt{P_{Y}(y)P_{X}(x)} \left(\biasdp(y) + \repm^\T\weightp(y)\right)\right.\\
  &\qquad\qquad \left.
    - \sqrt{P_{Y}(y)P_{X}(x)} o(\eps)\right]^2 + o(\eps^2)\\
&\overset{(*)}{=} \frac{1}{2} \sum_{x\in \X, y \in \Y} \left[\dtm(y, x) - \sqrt{P_{Y}(y)P_{X}(x)}\weightp^\T(y)\repp(x)\right]^2\\
&\qquad + \frac{1}{2} \sum_{x\in \X, y \in \Y} \left[\sqrt{P_{Y}(y)P_{X}(x)} \left(\biasdp(y) + \repm^\T\weightp(y)\right)\right]^2\\
&\qquad\qquad
+ o(\eps^2)\\
&=\frac{1}{2} \sum_{x\in \X, y \in \Y} \left[\dtm(y, x) - \bigl(\winf(y)\bigr)^\T\rinf(x)\right]^2 \\
 &\qquad
+ \frac{1}{2} \Ed{P_Y}{(\biasdp(y) + \repm^\T\weightp(y))^2} + o(\eps^2)\\
&= \frac{1}{2} \frob{\dtm - \wmat \bigl(\rmat\bigr)^\T}^2 + \frac{1}{2}\lossb^{ (\weight,\bias) }(\rep)  + o(\eps^2),\\
  \end{align*}
where to obtain $(*)$ we have used \eqref{local:1:2}-\eqref{local:1:3} together with the fact $|\dtm(y, x)| < \eps$, and 
\begin{gather*}
  \sum_{x\in \X, y \in \Y} \dtm(y, x)\sqrt{P_{Y}(y)P_{X}(x)}\left(\biasdp(y) + \repm^\T\weightp(y)\right) = 0,\\
  \sum_{x\in \X, y \in \Y} P_{Y}(y)P_{X}(x)\weightp^\T(y)\repp(x)\left(\biasdp(y) + \repm^\T\weightp(y)\right) = 0,
\end{gather*}
since 
$\E{\biasdp(Y)} = 0, \E{\repp(X)} = \E{\weightp(Y)} = 0$.

\end{proof}

\subsection{Proofs of Theorem~\ref{thm:softmax:ace:f} and Theorem~\ref{thm:softmax:ace:b}} \label{app:4}
\thmref{thm:softmax:ace:f} and \thmref{thm:softmax:ace:b} can be proved based on 
\lemref{lem:local:kl}.
\begin{proof}[Proofs of \thmref{thm:softmax:ace:f} and \thmref{thm:softmax:ace:b}]
Note that the value of $\biasd(\cdot)$ only affects the second term of the K-L divergence, hence we can always choose $\biasd(\cdot)$ such that $\biasdp(y) + \repm^\T\weightp(y) = 0$. Then the $(\wmat, \rmat)$ pair should be chosen as
\begin{equation}
  (\wmat, \rmat)^* = \argmin_{(\wmat, \rmat)} \bfrob{\dtm - \wmat \bigl(\rmat\bigr)^\T}^2.
\end{equation}
Set the derivative\footnote{In this paper, we use the denominator-layout notation of matrix calculus where the scalar-by-matrix derivative will have the same dimension as the matrix.}
\begin{equation}
  \frac{\p }{\p \wmat}\frob{\dtm - \wmat \bigl(\rmat\bigr)^\T}^2 = 2(\wmat\bigl(\rmat\bigr)^\T\rmat  - \dtm\rmat)
\end{equation}
 to zero, and the optimal $\wmat$ for fixed
$\rmat$ is\footnote{Here we assume the matrix $\bigl(\rmat)\bigr)^\T
  \rmat$, i.e., $\bLa_{\repp(X)}$ is invertible. For the case where
  $\bigl(\rmat\bigr)^\T \rmat$ is singular, all conclusions are the
  same when we use Moore-Penrose inverse to replace ordinary matrix
  inverse.} 
\begin{equation}
  \wmat{}^* = \dtm \rmat(\bigl(\rmat\bigr)^\T \rmat)^{-1}.
  \label{eq:9}
\end{equation}
As $\bone^\T\sqrt{\bP_Y}\, \dtm = 0$, we have
$\bone^\T\sqrt{\bP_Y}\,\wmat{}^* = 0$, which demonstrates that
$\wmat{}^*$ is a valid matrix for a zero-mean feature vector.

To express $\wmat{}^*$ of \eqref{eq:9} in the form of $\rep$ and $\weight$, we can make use of the correspondence between feature and information vectors. Note that, for a zero-mean feature function $f(X)$ with corresponding information vector $\phi$, we have the correspondence $\Ed{P_{X|Y}}{f(X)|Y}  \leftrightarrow \dtm \phi$ since the $y$-th element of information vector $\dtm \phi$ is given by
\begin{align*}
  &\eqspace\sum_{x \in \X} \dtm(y, x) \phi(x)\\
  &= \sum_{x \in \X} \frac{P_{XY}(x, y) - P_X(x)P_Y(y)}{\sqrt{P_X(x)P_Y(y)}} f(x) \sqrt{P_X(x)}\\
  &= \frac{1}{\sqrt{P_Y(y)}} \sum_{x \in \X} P_{XY}(x, y)f(x)\\
  &= \frac{1}{\sqrt{P_Y(y)}} \Ed{P_{X|Y}}{f(X)|Y = y}.
\end{align*}
Using similar methods, we can verify that $\bLa_{\repp(X)} = \bigl(\rmat\bigr)^\T \rmat$. As a result, \eqref{eq:9} is equivalent to
\begin{equation}
  \weightp^*(y) =  \Ed{P_{X|Y}}{\bLa_{\repp(X)}^{-1} \repp(X) \Bigm| Y=y}.
\end{equation}

By symmetry, the first two equations of \thmref{thm:softmax:ace:b} can be proved using the same method. To obtain the third equations of these two theorems, we need to minimize $\lossb^{ (\weight,\bias) }(\rep) = \Ed{P_Y}{(\repm^\T \weightp(Y) + \biasdp(Y))^2}$. When $\weightp$ and $\repm$ are fixed, the optimal $\biasdp$ is
\begin{equation}
  \biasdp^*(y) = -\repm^\T \weightp(Y),
  \label{eq:3}
\end{equation}
and the corresponding $\lossb^{ (\weight,\bias) }(\rep) = 0$.

When $\biasdp$ and $\weightp$ are fixed, we have
\begin{equation}
  \begin{aligned}
    &\eqspace\lossb^{ (\weight,\bias) }(\rep)\\
    &= \Ed{P_Y}{(\repm^\T \weightp(Y) + \biasdp(Y))^2}\\
    &= \repm^\T\bLa_{\weightp(Y)}\,\repm + 2\repm^\T\Ed{P_Y}{\weightp(Y)\biasdp(Y)}
    + \var(\biasdp(Y)).
  \end{aligned}
  \label{eq:lossb}
\end{equation}
Set $\frac{\p}{\p \repm} \lossb^{ (\weight,\bias) }(\rep) = 0$ and we obtain 
\begin{equation}
  \repm^* = -\bLa_{\weightp(Y)}^{-1} \Ed{P_Y}{\weightp(Y)\biasdp(Y)}.
  \label{eq:repm*}
\end{equation}
\end{proof}

\subsection{Proof of Theorem~\ref{thm:softmax:ace}} \label{app:5}
\begin{proof}
  From \lemref{lem:local:kl}, choosing the optimal $(\wmat, \rmat)$ is equivalent to solving the matrix factorization problem of $\dtm$. Since both $\wmat$ and $\rmat$ have rank no greater than $k$, from the Eckart-Young-Mirsky theorem \cite{eckart1936approximation}, the optimal choice of $\wmat \bigl(\rmat\bigr)^\T$ should be the truncated singular value decomposition of $\dtm$ with top $k$ singular values. As a result, $(\wmat, \rmat)^*$ are the left and right singular vectors of $\dtm$ corresponding to the largest $k$ singular values.
  
  The optimality of bias $\biasdp(y) = -\repm^\T \weightp(y)$ has already been shown in \appref{app:4}.
\end{proof}

\subsection{Proof of Theorem~\ref{thm:mlp:ace}} \label{app:7}
The following lemma is useful to prove Theorem~\ref{thm:mlp:ace}.
\begin{lemma}[Pythagorean theorem]
  Let $\hrmat{}^*$ be the optimal 
  matrix for given $\hwmat$ as
  defined in \eqref{eq:softmax_ace_b:1}. Then, 
  \begin{equation}
    \begin{aligned}
    &\bfrob{\dtm - \wmat \bigl(\hrmat\bigr)^\T }^2 - \bfrob{\dtm - \wmat \bigl(\hrmat{}^*\bigr)^\T }^2\\
    &\qquad= \bfrob{ \wmat\bigl(\hrmat{}^*\bigr)^\T  -  \wmat \bigl(\hrmat\bigr)^\T }^2.
  \end{aligned}
    \label{eq:tri_eq}
  \end{equation}
\label{lemma:pythagorean}
\end{lemma}

\begin{proof}[Proof of Lemma \ref{lemma:pythagorean}]
  Denote by $\langle \bU, \bV  \rangle$ the Frobenius inner product of matrices $\bU$ and $\bV$, i.e., $\langle \bU, \bV \rangle \defeq \tr(\bU^\T\bV)$, and we have
    \begin{align*}
      &\eqspace
      \left\langle \dtm - \hwmat\bigl(\hrmat{}^*\bigr)^\T,
      \hwmat\bigl(\hrmat\bigr)^\T\right\rangle \\ 
      &= \tr\left(\dtm\hrmat\bigl(\hwmat\bigr)^\T\right) -
      \tr\left(\hrmat{}^*\bigl(\hwmat\bigr)^\T\hwmat
      \bigl(\hrmat\bigr)^\T\right)\\ 
      &= \tr\left(\dtm\hrmat\bigl(\hwmat\bigr)^\T\right) - \tr\left( \dtm^\T \hwmat\bigl(\hrmat\bigr)^\T\right)\\
      &= 0.
    \end{align*}
  As a result,
    \begin{align*}
      &\eqspace\bfrob{\dtm  -  \hwmat \bigl(\hrmat\bigr)^\T}^2\\
      &= \bfrob{\dtm - \hwmat\bigl(\hrmat{}^*\bigr)^\T  +
        \bigl(\hwmat\bigl(\hrmat{}^*\bigr)^\T - \hwmat
        \bigl(\hrmat\bigr)^\T \bigr)}^2\\ 
      &= \bfrob{\dtm - \hwmat\bigl(\hrmat{}^*\bigr)^\T}  +  \bfrob{\hwmat\bigl(\hrmat{}^*\bigr)^\T - \hwmat \bigl(\hrmat\bigr)^\T}^2 \\
      &\qquad + 2\left\langle \dtm - \hwmat\bigl(\hrmat{}^*\bigr)^\T, \hwmat(\bigl(\hrmat{}^*\bigr)^\T - \bigl(\hrmat\bigr)^\T)\right\rangle\\
      &= \bfrob{\dtm - \hwmat\bigl(\hrmat{}^*\bigr)^\T}  +  \frob{\hwmat\bigl(\hrmat{}^*\bigr)^\T - \hwmat \bigl(\hrmat\bigr)^\T}^2.
  \end{align*}
\end{proof}
\begin{proof}[Proof of \thmref{thm:mlp:ace}]
From Lemma~\ref{lemma:pythagorean}, we have
\begin{align*}
  &\eqspace\loss (\hrep) - \loss (\hrep^*)\\
  &= \frac{1}{2} \left[\frob{ \dtm - \wmat \bigl(\hrmat\bigr)^\T }^2 - \frob{ \dtm - \wmat \bigl(\hrmat{}^*\bigr)^\T }^2\right]\\
  &\qquad
    + \frac{1}{2} \left[\lossb^{ (\weight,\bias) }(s) - \lossb^{ (\weight,\bias) }(s^*)\right] + o(\eps^2)\\
  &= \frac{1}{2} \frob{\wmat\bigl(\hrmat{}^*\bigr)^\T  -  \wmat \bigl(\hrmat\bigr)^\T}^2 + \frac{1}{2}\lossbh^{(\weight, \bias)} (\hrep, \hrep^*) + o(\eps^2),
\end{align*}
where $\lossbh^{(\weight, \bias)} (\hrep, \hrep^*) \defeq \lossb^{ (\weight,\bias) }(s) - \lossb^{ (\weight,\bias) }(s^*)$. We then optimize $ \frob{\wmat\bigl(\hrmat{}^*\bigr)^\T  -  \wmat \bigl(\hrmat\bigr)^\T}^2$ and $\lossbh^{(\weight, \bias)} (\hrep, \hrep^*)$ separately.


For the first term, we need to express $\hrmat$ in terms of $\bhwmatv$ and $\bhrmat$. From \eqref{eq:layer_s:2} we obtain
\begin{subequations}
  \begin{gather}
    \E{\hrep_z(X)} = \actfun\left(\bhbias(z)\right) + o(\eps),\label{eq:bhbias}\\
    \hrepp_z(x) = \bhweight^\T(z)\bhrepp(x) \cdot \actfun'\left(\bhbias(z)\right) + o(\eps),
  \end{gather}
\end{subequations}
which can be expressed in information vectors as
  \begin{equation}
    \hrmat = \bhrmat\bhwmatv^\T\bhjmat + o(\eps).
    \label{eq:8}
  \end{equation}
From \thmref{thm:softmax:ace:b}, we have
  \begin{equation}
    \hrmat{}^* = \dtm^\T\,\hwmat\,\bigl(\bigl(\hwmat\bigr)^\T\hwmat\bigr)^{-1}.
    \label{eq:10}
  \end{equation}
  As a result, we have
    \begin{align}      
      &\eqspace\bfrob{\hwmat\bigl(\hrmat{}^*\bigr)^\T - \hwmat\bigl(\hrmat\bigr)^\T}^2\notag\\
      &=
      \bfrob{\bigl(\bigl(\hwmat\bigr)^\T\hwmat\bigr)^{1/2}(\bigl(\hrmat{}^*\bigr)^\T
        - \bigl(\hrmat\bigr)^\T)}^2\notag\\ 
      &=
      \Big\| \bigl(\bigl(\hwmat\bigr)^\T\hwmat\bigr)^{1/2} 
      \cdot \left(\bigl(\hrmat{}^*\bigr)^\T
        - \bhjmat\bhwmatv\bigl(\bhrmat\bigr)^\T -
        o(\eps)\right) \Big\|_{\text{F}}^2 \notag\\ 
      &=
      \Big\| \bigl(\bigl(\hwmat\bigr)^\T\hwmat\bigr)^{1/2} 
      \cdot \left(\bigl(\hrmat{}^*\bigr)^\T - \bhjmat\bhwmatv\bigl(\bhrmat\bigr)^\T\right) \Big\|_{\text{F}}^2 + o(\eps^2)\notag\\
      &=
      \Big\| \bigl(\bigl(\hwmat\bigr)^\T\hwmat\bigr)^{1/2}\bhjmat
       \notag
      \\ 
      &\qquad
      \cdot \left(\bhjmat^{-1}\bigl(\hrmat{}^*\bigr)^\T 
        - \bhwmatv\bigl(\bhrmat\bigr)^\T\right) \Big\|_{\text{F}}^2 + o(\eps^2)\notag\\
      &=\bfrob{\prior \hdtm -  \prior\bhwmatv \bigl(\bhrmat\bigr)^\T}^2
        + o(\eps^2),
        \label{eq:6}
    \end{align}
where the third equality follows from the fact that [cf. \eqref{local:1:3}] $\hrepp(x) = O(\eps)$ and $\hweightp(y) = O(1)$, and the last equality follows from the definitions $\hdtm \defeq \bhjmat^{-1}\bigl(\hrmat{}^*\bigr)^\T$ and $\prior \defeq (\bigl(\wmat\bigr)^\T\wmat)^{1/2} \bhjmat$.
  
For the second term, from \eqref{eq:lossb} and \eqref{eq:repm*}, we have
\begin{align}
  &\eqspace
    \lossbh^{(\weight, \bias)} (\hrep, \hrep^*)\notag\\
  &= [(\hrepm - \hrepsm) + \hrepsm]^\T\bLa_{\weightp(Y)} \bigl[
    (\hrepm - \hrepsm) + \hrepsm \bigr]\notag\\
  &\qquad  -\hrepsm^\T\bLa_{\weightp(Y)}\hrepsm + 2(\hrepm - \hrepsm)^\T\Ed{P_Y}{\weightp(Y)\biasdp(Y)}\notag\\
  &= (\hrepm - \hrepsm)^\T\bLa_{\weightp(Y)} \bigl(\hrepm - \hrepsm\bigr)\notag\\
  &\qquad + 2 (\hrepm - \hrepsm)^\T\left(\bLa_{\weightp(Y)}\hrepsm + \Ed{P_Y}{\weightp(Y)\biasdp(Y)}\right)\notag\\
  &= (\hrepm - \hrepsm)^\T\bLa_{\weightp(Y)}\bigl(\hrepm - \hrepsm\bigr).    \label{eq:7}
\end{align}
  Combining \eqref{eq:6} and \eqref{eq:7} finishes the proof.  
\end{proof}

The results of $\hrepm^*$ and $\bhweight^*$ can be obtained via minimizing the loss function $\loss$. Again, these two terms can be optimized separately. To obtain $\hrepm^*$, consider the case where the hidden layer has used a bounded activation function, i.e., $\actfun_{\min} \preceq \hrepm \preceq \actfun_{\max}$, such as sigmoid function $1/(1 + e^{-x})$ or $\tanh(x)$. Then the optimal $\hrepm$ is the solution of
 \begin{equation}
   \begin{aligned}
     &\underset{\hrepm}{\text{minimize}} & &(\hrepm - \hrepsm)^\T\bLa_{\weightp(Y)}\bigl(\hrepm - \hrepsm\bigr)\\
     &\text{subject to} & &\actfun_{\min} \preceq \hrepm \preceq \actfun_{\max}.\\
   \end{aligned}
   \label{eq:optmize:hrepm}
 \end{equation}
 
 If $\hrepsm$ satisfies the constraint of \eqref{eq:optmize:hrepm}, then it is the optimal solution. Otherwise, some elements of $\hrepm^*$ will become either $\actfun_{\min}$ or $\actfun_{\max}$, which is known as the saturation phenomenon.


 Further, from \eqref{eq:bhbias}, the bias $\bhbias(z)$ of hidden layer is\footnote{When $\bhrepm \neq 0$, the formula should be modified as $\bhbias(z) = \actfun^{-1}(\hrepm^*(z)) - \bhrepm^\T\bhweight + o(\eps)$.}
 \begin{equation*}
   \bhbias(z) = \actfun^{-1}(\hrepm^*(z)) + o(\eps).
 \end{equation*}
 

To obtain $\bhwmatv{}^*$, let
 \begin{equation}
   \begin{aligned}
     \hdtm' &\defeq \prior \hdtm =
     \bigl(\bigl(\hwmat\bigr)^\T\hwmat\bigr)^{-1/2} \,
     \bigl(\hwmat\bigr)^\T \, \dtm,\\  
     \bhwmatv{}' &\defeq \prior \bhwmatv =
     \bigl( \bigl(\hwmat\bigr)^\T \hwmat \bigr)^{1/2}
     \bhjmat \bhwmatv,  
   \end{aligned}
 \end{equation}
 then the optimal $\bhwmatv{}'$ 
 is the solution of
 \begin{equation}
   \underset{\bhwmatv{}'}{\operatorname*{minimize}}\quad \frob{\hdtm{}' - \bhwmatv{}'\bigl(\bhrmat\bigr)^\T }^2,\\
 \end{equation}
 i.e., 
 \begin{equation}
   \bhwmatv'^{*} = \hdtm{}' \bhrmat(\bigl(\bhrmat\bigr)^\T\bhrmat)^{-1}. 
 \end{equation}
 Hence, $\bhwmatv{}^*$ is given by
 \begin{align*}
   \bhwmatv{}^*
   &= \prior^{-1}\bhwmatv'^{*}
     = \prior^{-1} \hdtm' \bhrmat(\bigl(\bhrmat\bigr)^\T\bhrmat)^{-1}\\
   &= \hdtm \bhrmat(\bigl(\bhrmat\bigr)^\T\bhrmat)^{-1}\\
   &= \bhjmat^{-1}\cdot [\hwmat(\bigl(\hwmat\bigr)^\T\hwmat)^{-1}]^\T\dtm\, \bhrmat(\bigl(\bhrmat\bigr)^\T\bhrmat)^{-1},
 \end{align*}
 where the term $\dtm\,\bhrmat(\bigl(\bhrmat\bigr)^\T\bhrmat)^{-1}$ corresponds to a feature projection of $\bhrepp(X)$: 
 \begin{equation}
  \dtm \, \bhrmat\,\bigl(\bigl(\bhrmat\bigr)^\T\bhrmat\bigr)^{-1}
  \leftrightarrow 
  \Ed{P_{X|Y}}{\bLa_{\bhrepp(X)}^{-1} \bhrepp(X)\Bigm| Y}.
 \end{equation}
 As a consequence, this multi-layer neural network is conducting a 
  generalized feature projection between features extracted from
 different layers. In practice problems, the projected feature
 $\Ed{P_{\bhrepp|Y}}{\bLa_{\bhrepp}^{-1} \bhrepp\middle|Y}$ only
 depends on the distribution $P_{\bhrepp|Y}$, and does not depend on
 the distribution $P_{X|Y}$. Therfore, the above computations can be
 accomplished without knowing the hidden random variable $X$ and
 can be applied to general cases.







\bibliographystyle{IEEEtran}
\bibliography{icml2018,learning}

\end{document}